\documentclass[a4paper,preprintnumbers,aps,amsmath,amssymb,prl,twocolumn,superscriptaddress,multiline]{revtex4-1}
\usepackage{xspace}
\usepackage{graphicx} 
\usepackage{latexsym} 
\usepackage[mathscr]{eucal} 
\usepackage{bm}
\usepackage{dcolumn}
\usepackage{amsthm}

\usepackage{color}

\newcommand{\comment}[1]{}

\newcommand{\ket}[1]{|#1\rangle}
\newcommand{\bra}[1]{\langle#1|}
\newcommand{\braket}[2]{\langle#1|#2\rangle}

\newcommand{\etal}{\textit{et al.}\xspace}
\newcommand{\ie}{\textit{i.e.}\xspace}
\newcommand{\eg}{\textit{e.g.}\xspace}

\newcommand{\Tr}{\mathrm{Tr}}

\newcommand{\tr}[1]{%
{}^t\hspace{-0.8mm}#1%
}%
\newcommand{\rank}[1]{\mathrm{rank}~#1}

\newcommand{\Hilb}[1]{\mathcal{H}_{#1}}
\newcommand{\res}{\mathrm{r}}
\newcommand{\IN}{\mathrm{in}}

\newcommand{\I}{\mathbb{I}}
\newcommand{\R}[1]{\vec{R}_{#1}}
\newcommand{\Eps}[1]{\mathscr{E}_{#1}}

\newcommand{\norm}[1]{\lVert#1\rVert}

\newtheorem{lemma}{Lemma}
\newtheorem{claim}{Claim}

\begin{document}

\title{
   Entanglement cost of implementing controlled-unitary operations
}

\author{Akihito Soeda}
\affiliation{Graduate School of Science, the University of Tokyo, 7-3-1, Hongo, Bunkyo-ku, Tokyo, Japan.} 
\affiliation{Centre for Quantum Technologies, National University of Singapore, 3 Science Drive 2, 117543 Singapore, Singapore.}
\author{Peter S. Turner}
\affiliation{Graduate School of Science, the University of Tokyo, 7-3-1, Hongo, Bunkyo-ku, Tokyo, Japan.}
\author{Mio Murao}
\affiliation{Graduate School of Science, the University of Tokyo, 7-3-1, Hongo, Bunkyo-ku, Tokyo, Japan.}
\affiliation{Institute for Nano Quantum Information Electronics, the University of Tokyo, 4-6-1, Komaba, Meguro-ku, Tokyo, Japan.}

\date{\today}

\begin{abstract}
We investigate the minimum entanglement cost of the deterministic implementation of two-qubit controlled-unitary operations using local operations and classical communication (LOCC).  We show that any such operation can be implemented by a three-turn LOCC protocol, which requires at least 1 ebit of entanglement when the resource is given by a bipartite entangled state with Schmidt number 2. Our result implies that there is a gap between the minimum entanglement cost and the entangling power of controlled-unitary operations.  This gap arises due to the requirement of implementing the operations while oblivious to the identity of the inputs.
\end{abstract}

\keywords{entanglement assisted LOCC implementation, minimal entanglement resource, LOCC feasibility}

\maketitle

The theoretical cost of a quantum computation sets the minimum experimental resource requirements for realizing that quantum computation in practice.  Entanglement cost is one important resource that needs to be minimized, especially in distributed quantum computation \cite{DQC}, where global unitary operations are implemented by local (quantum) operations and classical communication (LOCC) \cite{LOCC}  assisted by entanglement.

Generally,  {\it  entanglement-assisted LOCC implementation} of an operation involves a quantum system consisting of two parts, namely an \textit{input system} and a \textit{resource system}.  The input system may start in an arbitrary {\it unknown} state representing quantum information.  On the other hand, the resource system is set to a particular known state that may depend on the unitary operation to be implemented, but not on the input quantum information.  We take ``LOCC implementation" to imply ``deterministic entanglement-assisted LOCC implementation'' in the rest of this Letter.

The amount of entanglement $E$ of the resource system is called the \textit{minimum entanglement cost} of LOCC implementation of a {\it unitary} operation $U$ if the following two properties hold, (i) LOCC implementation is feasible using an entangled state with $E$, and (ii) LOCC implementation of $U$ is infeasible for any resource state with less than $E$.  Property (i) guarantees that $E$ is an upper bound on the minimum entanglement cost, which can be proved by providing an explicit construction of the implementation protocol.  On the other hand, property (ii) guarantees that $E$ is the corresponding lower bound, and proving it is considerably more difficult, because one must show that the implementation fails for \textit{any} LOCC protocol.

Various attempts have been made to minimize $E$.  In \cite{Eisert}, Eisert \etal\ discovered a protocol that implements any two-qubit \textit{controlled}-unitary operation using LOCC assisted by a two-qubit maximally entangled state, or 1 {\it ebit} of entanglement.  Hence, for two-qubit controlled-unitary operations, the minimum entanglement is upper bounded by 1 ebit.
Since the maximum amount of entanglement created by performing a unitary operation, called its entangling power \cite{entanglingpower}, cannot exceed its minimum entanglement cost, we can conclude that the minimum entanglement cost of a controlled-unitary operation is 1 ebit if its entangling power is also 1 ebit.  A controlled-NOT operation and its local unitary equivalents are examples.  However, the minimum entanglement cost for general controlled-unitary operations has been an open question for a decade.

In order to estimate the minimum entanglement cost, let us consider the case where the input state for the LOCC implementation of $U$ is a pure product state.  If the identity of the input is also provided at the beginning of the implementation (but only after the preparation of the resource state), then we can employ an input-dependent LOCC implementation protocol.  In this case, LOCC implementability reduces to LOCC convertibility~\cite{Majorization} between the two joint states of the input and resource system before and after performing $U$.  If $U$ is a two-qubit controlled-unitary operation denoted by $U_u$, it can generate states with a Schmidt number up to 2.  For any given two pure entangled states with Schmidt number 2, one can be converted to the other by LOCC if and only if the former is more entangled than the latter. Hence, if the resource state has as much entanglement as the entangling power of $U_u$ and Schmidt number 2, the LOCC implementation of $U_u$ on a known pure product state is possible.

Because the entangling power of $U_u$ is less than 1 ebit in general~\cite{entanglingpower}, this might make it natural to expect that $U_u$ can be implemented using less than 1 ebit of entanglement resource with Schmidt number 2, even if the identity of the state remains unknown. Indeed, when the deterministicity condition has been relaxed, it is known that there is a situation where the average entanglement consumption can be reduced below 1 ebit in the asymptotic limit \cite{Cirac}.

In this Letter, however, we prove that for deterministic cases,  LOCC implementation for any given two-qubit controlled-unitary operation on \textit{unknown} inputs requires \textit{at least} 1 ebit irrespective of its entangling power, when the resource is given by a bipartite entangled state with Schmidt number 2.  Our result answers this long open question in entanglement theory.

The proof proceeds in the following four steps.  First, we show that any LOCC implementation protocol must be such that the number of each party's local measurements, or {\it turns}, is greater than or equal to 3.  Second, we present reductions of the descriptions of controlled-unitary operations and resource entanglement, using local unitary equivalence.  Third, we show that any $(n>3)$-turn LOCC implementation of a controlled-unitary operation can be transformed to an $(n-1)$-turn protocol by investigating three cases that cover all possible LOCC implementation protocols for controlled-unitary operations.  Finally, by induction, we see that an $n$-turn implementation protocol can be converted to a 3-turn one, for which it is proved that the amount of entanglement of the resource state must be 1 ebit.

We start by describing a LOCC implementation mathematically.  $\Hilb{\IN}$ and $\Hilb{\res}$ will be used to denote the Hilbert space of the input system and the resource system, respectively.     We use $\{ \ket{0}, \ket{1} \}$ as a basis for a qubit Hilbert space.  To specify the Hilbert space where a given state belongs, we use the corresponding subscripts in both vector and operator notations, \eg, $\ket{0}_{X} \in \Hilb{X}$ and $\rho_{X} \in \mathcal{S}(\Hilb{X})$ for the Hilbert space $\mathcal{H}_X$.  A LOCC implementation of a unitary operation $U$ on $\Hilb{\IN}$ assisted by a given entangled state represented by $\rho_\res $ is a completely positive and trace-preserving (CPTP) map $\Gamma: \mathcal{S} (\Hilb{\IN} \otimes \Hilb{\res}) \longrightarrow \mathcal{S} (\Hilb{\IN})$, such that it is implementable by LOCC and satisfies
\begin{equation} \label{implementation}
 \Gamma(\rho_\IN \otimes \rho_\res) = U \rho_\IN U^\dag
\end{equation}
for all input states $\rho_\IN$.  The state $\rho_\res$ is a resource state for such a LOCC implementation of $U$.

The input system consists of two parts, namely, Alice's input qubit and Bob's input qubit, whose corresponding Hilbert spaces are denoted by $\Hilb{A,\IN}$ and $\Hilb{B,\IN}$, respectively.  The resource system also consists of a bipartite system shared between Alice and Bob, where the corresponding Hilbert spaces are denoted by $\Hilb{A,\res}$ and $\Hilb{B,\res}$.  We call $\Hilb{A,\IN} \otimes \Hilb{A,\res}$ \textit{Alice's subsystem} and $\Hilb{B,\IN} \otimes \Hilb{B,\res}$ \textit{Bob's subsystem}, where $\Hilb{A}$ and $\Hilb{B}$ are used to denote the Hilbert space corresponding to each party's subsystem.  We will sometimes abbreviate $\ket{k}_{A,\IN} \otimes \ket{l}_{B,\IN}$ as $\ket{kl}_{\IN}$ and $\ket{k}_{A,\res} \otimes \ket{l}_{B,\res}$ as $\ket{kl}_{\res}$.

At each turn in a general two-party LOCC protocol, either Alice or Bob performs any generalized measurement operation, which is described by two sets of measurement operators $\{ M^{(r)} \}_r$  for Alice and $\{ K^{(r)} \}_r$  for Bob, where $r$ denotes a measurement outcome on each subsystem, and then classically communicates $r$ to the other party.  Note that unitary operations on each subsystem are included as a special case where $r$ has only one value.  Since the only constraint on the set of measurement operators is the completeness relation, and operations at each turn can depend arbitrarily on the previous turns, the set of possible generalized measurement operations in LOCC is very large.

In order to manage this set, we focus on the accumulated effect that is brought on by successive operations in a given LOCC protocol.  Let $r_k$ denote the measurement outcome of the $k$-th turn, and $\R{k} = (r_1, r_2, \ldots , r_k)$ denote the sequence of measurement outcomes of the first $k$ turns.  The $k$th generalized measurement operation is a function of $\R{k-1}$.  With a slight abuse of notation, we set
$
 \R{k} = (\R{k-1},r_k) = (\R{k-2},r_{k-1},r_k) = \ldots
$
and so on.

We use $M^{(r_k|\R{k-1})}$ and $K^{(r_k|\R{k-1})}$ to denote Alice's and Bob's measurement operators at the $k$th turn, respectively.  To make the definition consistent, if the $k$-th turn is not Alice's turn to perform an operation, we set
$M^{(r_k|\R{k-1})} = \I_A$, where $\I_A$ is the identity operator on $\Hilb{A}$.  Bob's measurement operators are defined similarly.

We define \textit{accumulated operators} $A^{\R{k}}$ and $B^{\R{k}}$ that represent the accumulated effect of the measurement operators
 after $k$ turns on Alice's and Bob's qubits, respectively, by
$A^{\R{k}} = M^{(r_k|\R{k-1})}M^{(r_{k-1}|\R{k-2})} \cdots M^{(r_1|\R{0})}$
and
$B^{\R{k}} = K^{(r_k|\R{k-1})}K^{(r_{k-1}|\R{k-2})} \cdots K^{(r_1|\R{0})}$,
where $M^{(r_1|\R{0})} = M^{(r_1)}$ and $K^{(r_1|\R{0})} = K^{(r_1)}$.  The accumulated operators naturally form a tree graph where each set of outcomes $\R{k}$ labels a \textit{branch} of length $k$.

The lengths of the branches in a LOCC protocol are not necessarily the same, but Alice and Bob can repeatedly perform identity operations to fill the shorter ones so that all branches have the same length $n$.  Hence, we only need to consider protocols such that the target operation $U$ is implemented after $n$ turns.  With this reduction and rewriting the CPTP map using accumulated operators $A^{\R{n}}$ and $B^{\R{n}}$, Eq.\ (\ref{implementation}) now reads
$
 \sum_{\R{n}} (A^{\R{n}} \otimes B^{\R{n}}) (\rho_\IN \otimes \rho_\res) (A^{\R{n}} \otimes B^{\R{n}})^\dag = U \rho_\IN U^\dag.
$

Because $U$ is a unitary operation, $U \rho_\IN U^\dag$ is a pure state when $\rho_\IN$ is pure.  Thus, for each term in the summation, deterministicity requires a positive coefficient $p^{\R{n}}$ exist such that
\begin{equation}\label{implementation_single_ao}
 (A^{\R{n}} \otimes B^{\R{n}}) (\rho_\IN \otimes \rho_\res) (A^{\R{n}} \otimes B^{\R{n}})^\dag = p^{\R{n}} U \rho_\IN U^\dag.
\end{equation}
From this equation, we can see that a mixed state $\rho_\res = \sum_k p_k \ket{\psi_k}\bra{\psi_k}$ is a resource state if and only if each of $\ket{\psi_k}\bra{\psi_k}$ can be used as a resource.  From here on, we only consider pure states for the entanglement resource.

As the first step of our proof, we show that the number of turns, $n$, must be greater than or equal to 3.  It is shown in \cite{cccapacity} that if a global unitary operation $U$ is implementable by a LOCC protocol between Alice and Bob, the same protocol can be used for sending nonzero classical information from Alice to Bob or vice versa by choosing an appropriate input state for each party. Therefore, if LOCC implementation of $U$ with $n=2$ turns is possible, a one-way LOCC protocol from Alice to Bob should be possible to send classical information from Bob to Alice. However, such a one-way LOCC protocol cannot change Alice's reduced density matrix depending on Bob's input state, hence, classical communication from Bob to Alice is impossible. Thus, taking the contrapositive, LOCC implementation of $U$ with $n=2$ turns is impossible.  Impossibility with $n=1$ protocols is trivial.

The argument so far holds for LOCC implementation of arbitrary (nontrivial) two-qubit unitary operations.  We proceed to the second step of our proof by specifying  them to be controlled-unitary operations.  A general form is then given by $U_u=v_1 \otimes v_2 \cdot (\ket{0} \bra{0} \otimes \I + \ket{1} \bra{1} \otimes u) \cdot w_1 \otimes w_2$, where $v_1, v_2, u, w_1,w_2$ are single-qubit local unitary operations.  $U_u$ can be further locally transformed to
\begin{equation} \label{eq:LOCCimp_Utheta}
U_\theta=\ket{00}\bra{00}+\ket{01}\bra{01}+\ket{10}\bra{10}+e^{i\theta}\ket{11}\bra{11}
\end{equation}
where $\theta$ is a nonlocal parameter,  by performing appropriate local unitary operations that depend only upon $U_u$.  The local unitary operations taking $U_u$ to $U_\theta$ are known to both Alice and Bob, so without loss of generality we only need to consider protocols implementing the single-parameter family $U_\theta$.

In this step, we also specify our resource to be a bipartite entangled pure state with Schmidt number (the number of nonzero coefficients in its Schmidt decomposition) 2.   Since any such state can be transformed into a two-qubit entangled state by local unitary operations, we can reduce the form of this resource state to $ \ket{\res}_\res = (X_{A,\res} \otimes \I_{B,\res}) \ket{\Phi^+}_\res$, where $X \equiv \sqrt{\mu} \ket{0}\bra{0}+\sqrt{(1-\mu)} \ket{1}\bra{1}$ is a Choi matrix~\cite{Choi}  and $\ket{\Phi^+}_\res \equiv \ket{00}_\res + \ket{11}_\res$, by taking an approriate choice of basis.
With these notations, Eq.\ (\ref{implementation_single_ao}) now reads
\begin{equation}\label{eq:AccuAction}
(A^{\R{n}} \otimes B^{\R{n}}) (\ket{ij}_\IN \otimes \ket{\res}_\res) = c^{\R{n}} (U_\theta \ket{ij}_\IN)
\end{equation}
for all $\ket{ij}_\IN$ and some input-independent $c^{\R{n}} \in \mathbb{C}$ such that $|c^{\R{n}}|^2 = p^{\R{n}}$ for all $\R{n}$.

We now proceed to the third and most demanding step of our proof.   As $A^{\R{n}}$ and $B^{\R{n}}$ are operators taking a four-dimensional Hilbert space to a two-dimensional Hilbert space, by introducing unnormalized states $\ket{a^{\R{n}}_{ki}}_{A,\res}$ and $\ket{b^{\R{n}}_{lj}}_{B,\res}$, we have
\begin{eqnarray}
 A^{\R{n}} &=& \sum_{k,i=0,1} \ket{k}_{A,\IN}\bra{i} \otimes \bra{a^{\R{n}}_{ki}}_{A,\res},\\
 B^{\R{n}} &=& \sum_{l,j=0,1} \ket{l}_{B,\IN}\bra{j} \otimes \bra{b^{\R{n}}_{lj}}_{B,\res}.
\end{eqnarray}
For clarity, the subscripts of states will be dropped from here on.  Equation (\ref{eq:AccuAction}) now implies the following conditions,
\begin{equation}\label{eq:Blocks}
\bra{a^{\R{n}}_{ki}} X \ket{b^{\R{n}*}_{lj}} = \left\{ \begin{array}{ll}
c^{\R{n}} \bra{i j} U_\theta \ket{i j} & k=i \,\mathrm{and}\, l=j\\
0 & \mathrm{otherwise}, \end{array} \right.
\end{equation}
where we have used $\ket{b^{\R{n}*}_{lj}}$ to denote the complex conjugate of $\ket{b^{\R{n}}_{lj}}$ in the $\{ \ket{0}, \ket{1} \}$ basis.  The vectors $\bra{a_{00}^{\R{n}}}$ and $\bra{a_{11}^{\R{n}}}$ must be linearly independent because $e^{i \theta} \neq 1$, and for the same reason so must $\bra{b_{00}^{\R{n}}}$ and $\bra{b_{11}^{\R{n}}}$ be.  It is easy to show that $\bra{a_{ki}^{\R{n}}}$ and $\bra{b_{lj}^{\R{n}}}$ must be zero-vectors if $k \neq i$ and $l \neq j$, respectively.  Therefore, $A^{\R{n}}$ and $B^{\R{n}}$ have the following forms
\begin{align}
 \label{eq:ARn} A^{\R{n}} &= \ket{0}\bra{0} \otimes \bra{a^{\R{n}}_{00}} + \ket{1}\bra{1} \otimes \bra{a^{\R{n}}_{11}} ,\\
 \label{eq:BRn} B^{\R{n}} &= \ket{0}\bra{0} \otimes \bra{b^{\R{n}}_{00}} + \ket{1}\bra{1} \otimes \bra{b^{\R{n}}_{11}}.
\end{align}

Without loss of generality we now consider the case where Bob performs the final ($n$-th) turn in an implementation protocol.  As defined above, the accumulated operators for the previous three turns then satisfy the following relations:
\begin{eqnarray}
 \label{eq:LOCCimp_onA1} A^{\R{n-1}} &=& A^{\R{n}} =\; A^{(\R{n-1},r_n)},\\
 \label{eq:LOCCimp_onA2} A^{\R{n-2}\dag} A^{\R{n-2}} &=&\sum_{r_{n-1}} A^{(\R{n-2},r_{n-1})\dag} A^{(\R{n-2},r_{n-1})}, \\
 \label{eq:LOCCimp_onA3} A^{\R{n-3}} &=& A^{\R{n-2}} =\; A^{(\R{n-3},r_{n-2})},
\end{eqnarray}
and
\begin{eqnarray}
 \label{eq:LOCCimp_onB1} B^{\R{n-1}\dag} B^{\R{n-1}} &=& \sum_{r_n} B^{(\R{n-1},r_n)\dag} B^{(\R{n-1},r_n)} ,\\
 \label{eq:LOCCimp_onB2} B^{\R{n-2}} &=& B^{\R{n-1}} =\; B^{(\R{n-2},r_{n-1})},\\
 \label{eq:LOCCimp_onB3} B^{\R{n-3}\dag} B^{\R{n-3}} &=& \sum_{r_{n-2}} B^{(\R{n-3},r_{n-2})\dag} B^{(\R{n-3},r_{n-2})}.
\end{eqnarray}

We show that any $(n\geq 3)$-turn LOCC protocol implementing $U_\theta$ can be transformed into a $(n-1)$-turn LOCC protocol by investigating three types of transformations conditioned on the accumulated operators that arise at the $(n-3)$-th turn, $A^{\R{n-3}}$ and $B^{\R{n-3}}$.  A rigorous proof for all three cases requires intensive analysis of the relations that hold for the accumulated operators, and is given in the supplementary material \cite{supplemental}.  In what follows, we concentrate on that which is important for understanding the three transformations.

Defining the block elements  $A_{kk}^{\R{n-3}}$ and $B_{ll}^{\R{n-3}}$ of the accumulated operators $A^{\R{n-3}}$ and $B^{\R{n-3}}$ by
\begin{eqnarray}
 \label{fullA} A_{kk}^{\R{n-3}} &=& \sqrt{\sum_{r_{n-1}} \ket{a^{\R{n}}_{kk}}\bra{a^{\R{n}}_{kk}}},\\
 \label{fullB} B_{ll}^{\R{n-3}} &=& \sqrt{\sum_{r_{n-2},r_n} \ket{b^{\R{n}}_{ll}}\bra{b^{\R{n}}_{ll}}},
\end{eqnarray}
and using Eqs.\ (\ref{eq:ARn}) to (\ref{eq:LOCCimp_onB3}), we obtain the following relations
\begin{eqnarray} \label{eq:ARn-3}
  A^{\R{n-3}\dag} A^{\R{n-3}} &=& \sum_{k=0,1} \ket{k}\bra{k} \otimes (A_{kk}^{\R{n-3}})^2, \\
 B^{\R{n-3}\dag} B^{\R{n-3}} &=& \sum_{l=0,1} \ket{l}\bra{l} \otimes (B_{ll}^{\R{n-3}})^2.
\end{eqnarray}
The ranks of $A_{kk}^{\R{n-3}}$ and $B_{ll}^{\R{n-3}}$ cannot be  taken independently for successful LOCC implementable protocols.  Indeed, we have the following lemma, whose proof is given in \cite{supplemental}.
\begin{lemma} \label{lem:equalranks}
The ranks of  $A_{00}^{\R{n-3}}$ and $ A_{11}^{\R{n-3}}$ must be the same for any successful protocol.  Additionally, if the rank of $A_{00}^{\R{n-3}}$ is 2, the ranks of $B_{00}^{\R{n-3}}$ and $B_{11}^{\R{n-3}}$ must be the same.
\end{lemma}
This implies that all successful protocols can be classified into the following three cases;  (a) $\rank{A_{00}^{\R{n-3}}} = 1$ (b) $\rank{A_{00}^{\R{n-3}}} = 2$,  $\rank{B_{00}^{\R{n-3}}}=1$   (c) $\rank{A_{00}^{\R{n-3}}} = 2$,  $\rank{B_{00}^{\R{n-3}}}=2$.

For cases (a) and (b), we have \cite{supplemental},
\begin{lemma} \label{lem:rank1}
If $\rank{A_{00}^{\R{n-3}}} = 1$, then $\{M^{(r_{n-1}|\R{n-2})}\}_{r_{n-1}}$ is simulateable by a random unitary operation $\{ p^{\R{n-1}}, U^{\R{n-1}} \}$.  On the other hand, if $\rank{B_{00}^{\R{n-3}}} = 1$, then $\{K^{(r_{n-2}|\R{n-3})}\}_{r_{n-2}}$ can be simulated by a random unitary operation $\{ q^{\R{n-2}}, V^{\R{n-2}} \}$.
\end{lemma}
When the operation at a given turn is a random unitary operation, the ``outcome''  of the operation $r_{n-1}$ for case (a) [or $r_{n-2}$ for case (b)] can be {\it chosen} by Bob for case (a) [Alice for case (b)] and can be communicated with the measurement outcome $r_{n-2}$ (or $r_{n-3}$).   Then the communication of $r_{n-1}$ from Alice to Bob (or $r_{n-2}$  from Bob to Alice) is no longer necessary. Therefore, the number of turns can be decreased from $n$ to $n-1$.

 For case (c),  we have \cite{supplemental,Andersson},
\begin{lemma} \label{lem:whenrankA2}
 If $\rank{A_{00}^{\R{n-3}}} = \rank{B_{00}^{\R{n-3}}} = 2$, then it is possible to replace $\{ K^{(r_{n-2}|\R{n-3})}\}_{r_{n-2}}$, $\{ M^{(r_{n-1}|\R{n-2})}\}_{r_{n-1}}$, and $\{ K^{(r_{n}|\R{n-1})}\}_{r_n}$ with $\{ {M'}^{(r_{n-2}|\R{n-3})}\}_{r_{n-2}}$, $\{ {K'}^{(r_{n-1}|\R{n-2})}\}_{r_{n-1}}$, and $\{ {M'}^{(r_{n}|\R{n-1})}\}_{r_n}$, without changing $A^{\R{m}}$ and $B^{\R{m}}$ for $m \leq n-3$.
\end{lemma}
After this replacement, notice that the $(n-3)$th turn and $(n-2)$th turn are both performed by Alice.  These turns can be combined into a single operation, reducing the total number of turns from $n$ to $n-1$.

Hence, for all three types of successful implementation protocols with $n \geq 3$ the total number of turns can be decreased by one.  By induction, all implementation protocols can be transformed to one with three turns.

In the final step of our proof, we show the necessity of a 1-ebit resource.  The following lemma is proved in \cite{supplemental}.
\begin{lemma} \label{lem:dke}
Suppose $n \geq 3$.  If $\rank{A_{00}^{\R{n-3}}} = 2$ and $A_{11}^{\R{n-3}}(A_{00}^{\R{n-3}})^{-1}$ is a unitary operation, then
\begin{equation} \label{eq:dnk}
 \sum_l A_{00}^{\R{n-3}} X \tr{B}_{ll}^{\R{n-3}} \cdot (c.c.) \propto \I.
\end{equation}
\end{lemma}

For a 3-turn LOCC protocol, $A^{\R{n-3}}$ and $B^{\R{n-3}}$ are the identity operator by definition.  Clearly, Lemma \ref{lem:dke} can be applied and Eq.\ (\ref{eq:dnk}) now reads,
$\sum_{l} X^2 \propto I$.  Because $X$ is a positive matrix, we see that $X \propto I$.  By the definition of $X$, we have that the resource state should be given by $\ket{\res} = \ket{\Phi^+}/\sqrt{2}$, which has 1 ebit of entanglement. 

In this letter, we analyzed deterministic entanglement-assisted LOCC implementation of two-qubit controlled-unitary operations and showed that any given two-qubit controlled-unitary operation can be implemented by a three-turn protocol, which requires at least 1 ebit of entanglement when the resource is given by a bipartite entangled state with Schmidt number 2.   Our result implies that such a protocol necessarily consumes more entanglement than it can create, raising interesting questions about connections to irreversibility.  This gap between the minimum entanglement cost and entangling power arises due to the requirement of implementation without knowing the inputs, since entanglement cost can be reduced by implementing an input-dependent protocol.

  Our result also indicates that, since it is possible to realize a controlled-unitary operation by composing several controlled-unitary operations with less entangling power, such a decomposition consumes much more entanglement than when the target controlled-unitary operation is directly implemented.  Finally, our proofs are constructive, in that we explicitly give the new protocol that achieves the implementation in fewer steps.

{\it Acknowledgments:}  This work is supported by the Project for Developing Innovation Systems of MEXT, Japan, the Global COE Program, MEXT, Japan, and JSPS by KAKENHI (Grant No. 23540463).   After this work was completed, complementary results \cite{Stahlke} showed that the lower bound here can be generalized to higher dimensions.   The authors also numerically found that by using resource states with Schmidt number greater than 2, it is possible to perform $U_u$ with less than 1 ebit of entanglement.  We thank the authors of this reference for bringing it to our attention.

\appendix

\section{Proof of Lemma~\ref{lem:equalranks}}

In this section, we provide a proof of Lemma \ref{lem:equalranks}, which is divided into two subsections, the first for the $A$ operators and the second for the $B$ operators.  The calculations presented here are also used in the subsequent proofs.

\subsection{The proof for $\rank{A_{00}^{\R{n-3}}}=\rank{A_{11}^{\R{n-3}}}$} \label{subsec:rankAequal}
 In this section, we show the ranks of the block elements of accumulated operators $A_{kk}^{\R{n-3}}$ for $k=0$ and $k=1$ should be same for successful LOCC implementation of controlled-unitary operations.

For any given $\R{n}$, $c^{\R{n}}$, and a full rank $X$, there always exists a unique pair of linearly independent vectors $\{ \bra{b_{00}^{\R{n}}}, \bra{b_{11}^{\R{n}}} \}$  corresponding to the linearly independent vectors $\{ \bra{a_{00}^{\R{n}}}, \bra{a_{11}^{\R{n}}} \}$.  To see this, first recall that $\bra{a_{00}^{\R{n}}}$ and $\bra{a_{11}^{\R{n}}}$ are linearly independent, therefore, (due to their Hilbert spaces being isomorphic), it is possible to decompose $\ket{b_{00}^{\R{n} *}}$ and $\ket{b_{11}^{\R{n}*}}$ as
\[
 \ket{b_{00}^{\R{n}*}} = x^{\R{n}} \ket{a_{00}^{\R{n}}} + y^{\R{n}} \ket{a_{11}^{\R{n}}}
\]
and
\[
 \ket{b_{11}^{\R{n}*}} = \xi^{\R{n}} \ket{a_{00}^{\R{n}}} + \eta^{\R{n}} \ket{a_{11}^{\R{n}}}.
\]

Substituting these equations into Eq.\ (\ref{eq:Blocks}), we obtain
\begin{align*}
 x^{\R{n}} \bra{a_{00}^{\R{n}}}X\ket{a_{00}^{\R{n}}} + y^{\R{n}} \bra{a_{00}^{\R{n}}}X\ket{a_{11}^{\R{n}}} &= c^{\R{n}}, \\
 x^{\R{n}} \bra{a_{11}^{\R{n}}}X\ket{a_{00}^{\R{n}}} + y^{\R{n}} \bra{a_{11}^{\R{n}}}X\ket{a_{11}^{\R{n}}} &= c^{\R{n}} ,\\
 \xi^{\R{n}} \bra{a_{00}^{\R{n}}}X\ket{a_{00}^{\R{n}}} + \eta^{\R{n}} \bra{a_{00}^{\R{n}}}X\ket{a_{11}^{\R{n}}} &= c^{\R{n}}, \\
 \xi^{\R{n}} \bra{a_{11}^{\R{n}}}X\ket{a_{00}^{\R{n}}} + \eta^{\R{n}} \bra{a_{11}^{\R{n}}}X\ket{a_{11}^{\R{n}}} &= c^{\R{n}} e^{i \theta}.
\end{align*}
Defining a matrix $L^{\R{n}} $ by
\[
L^{\R{n}} = \begin{pmatrix} \bra{a_{00}^{\R{n}}}X\ket{a_{00}^{\R{n}}} & \bra{a_{00}^{\R{n}}}X\ket{a_{11}^{\R{n}}}\\ \bra{a_{11}^{\R{n}}}X\ket{a_{00}^{\R{n}}} & \bra{a_{11}^{\R{n}}}X\ket{a_{11}^{\R{n}}}  \end{pmatrix},
\]
the four equations above are equivalent to
\begin{align}
 L^{\R{n}} \cdot \begin{pmatrix} x^{\R{n}} \\ y^{\R{n}} \end{pmatrix} &= c^{\R{n}} \begin{pmatrix} 1 \\ 1 \end{pmatrix}, \\
 L^{\R{n}} \cdot \begin{pmatrix} \xi^{\R{n}} \\ \eta^{\R{n}} \end{pmatrix} &= c^{\R{n}} \begin{pmatrix} 1 \\ e^{i \theta} \end{pmatrix}.
\end{align}
We denote the elements of $(L^{\R{n}})^{-1}$ as
\[
 (L^{\R{n}})^{-1} = \begin{pmatrix} (L^{\R{n}})^{-1}_{00} & (L^{\R{n}})^{-1}_{01} \\ (L^{\R{n}})^{-1}_{10} & (L^{\R{n}})^{-1}_{11} \end{pmatrix}.
\]
With this notation, we have
\begin{align*}
 x^{\R{n}} &= c^{\R{n}} ((L^{\R{n}})^{-1}_{00} + (L^{\R{n}})^{-1}_{01}),\\
 y^{\R{n}} &= c^{\R{n}} ((L^{\R{n}})^{-1}_{10} + (L^{\R{n}})^{-1}_{11},)\\
 \xi^{\R{n}} &= c^{\R{n}} ((L^{\R{n}})^{-1}_{00} + e^{i\theta}(L^{\R{n}})^{-1}_{01}),\\
 \eta^{\R{n}} &= c^{\R{n}} ((L^{\R{n}})^{-1}_{10} + e^{i\theta}(L^{\R{n}})^{-1}_{11}).
 \end{align*}
Note that the right hand side is determined by $\{ \bra{a_{00}^{\R{n}}},\bra{a_{11}^{\R{n}}}, c^{\R{n}}, X \}$, therefore $\bra{{b}_{00}^{\R{n}}}$ and $\bra{{b}_{11}^{\R{n}}}$ are unique functions of $\{ \bra{a_{00}^{\R{n}}},\bra{a_{11}^{\R{n}}}, c^{\R{n}}, X \}$.

This uniqueness of $\bra{b_{ll}^{\R{n}}}$ to $\{ \bra{a_{kk}^{\R{n}}}, c^{\R{n}}, X \}$ implies that the dependence of $\bra{b_{ll}^{\R{n}}}$ of $r_n$ is only through $c^{\R{n}}$.  To see this, for any given $\R{n-1}$ and $r_n$, Eq.\ (\ref{eq:LOCCimp_onA1}) implies that $A^{(\R{n-1},r_n)}$ does not have any $r_n$-dependence which means that $\bra{a_{00}^{\R{n}}}$ and $\bra{a_{11}^{\R{n}}}$ should also not have this $r_n$-dependence, \ie ,
\begin{align}
 \label{pqi} \bra{a_{00}^{(\R{n-1},r_n)}} &= \bra{a_{00}^{(\R{n-1},r'_n)}} ,\\
 \label{pqj} \bra{a_{11}^{(\R{n-1},r_n)}} &= \bra{a_{11}^{(\R{n-1},r'_n)}}
\end{align}
for any $r_n$ and $r'_n$.  We define $\bra{a_{00}^{\R{n-1}}}$ and $\bra{a_{11}^{\R{n-1}}}$ by
\begin{align}
 \label{a00Rn-1} \bra{a_{00}^{\R{n-1}}} &= \bra{a_{00}^{(\R{n-1},0)}} ,\\
 \label{a11Rn-1} \bra{a_{11}^{\R{n-1}}} &= \bra{a_{11}^{(\R{n-1},0)}}
\end{align}
Using this notation, we have, similar to Eq.(\ref{eq:ARn}),
\begin{equation} \label{ARn-1}
 A^{\R{n-1}} = \ket{0}\bra{0} \otimes \bra{a_{00}^{\R{n-1}}} + \ket{1}\bra{1} \otimes \bra{a_{11}^{\R{n-1}}}.
\end{equation}
  
It must be that $L^{\R{n}}$, whose $r_n$-dependence is only through $\bra{a_{00}^{\R{n}}}$ and $\bra{a_{11}^{\R{n}}}$, is independent of $r_n$, from which we learn that
\begin{align}
 \label{b00} \bra{b_{00}^{(\R{n-1},r_n)}} &= \frac{c^{(\R{n-1},r_n)}}{c^{(\R{n-1},r'_n)}} \bra{b_{00}^{(\R{n-1},r'_n)}},\\
 \label{b11} \bra{b_{11}^{(\R{n-1},r_n)}} &= \frac{c^{(\R{n-1},r_n)}}{c^{(\R{n-1},r'_n)}} \bra{b_{11}^{(\R{n-1},r'_n)}}.
\end{align}
Therefore, for two different outcomes $r_n$ and $r'_n$, $\bra{b_{ll}^{(\R{n-1},r_n)}}$ and $\bra{b_{ll}^{(\R{n-1},r'_n)}}$ are collinear.

Defining $\gamma^{\R{n}}$ by
\begin{equation} \label{def:gamma}
 \gamma^{\R{n}} = \frac{c^{(\R{n-1},{r_n})}}{c^{(\R{n-1},{0})}},
\end{equation}
and using Eq.\ (\ref{eq:LOCCimp_onB1}), we see that
\begin{align*}
 & B^{\R{n-1}\dag} B^{\R{n-1}} = \sum_{r_n} B^{(\R{n-1},r_n)\dag} B^{(\R{n-1},r_n)}\\
 &= \sum_{r_n} \ket{0}\bra{0} \otimes |\gamma^{(\R{n-1},r_n)}|^2 \ket{b_{00}^{(\R{n-1},0)}}\bra{b_{00}^{(\R{n-1},0)}} \\
 & \quad + \ket{1}\bra{1} \otimes |\gamma^{(\R{n-1},r_n)}|^2 \ket{b_{11}^{(\R{n-1},0)}}\bra{b_{11}^{(\R{n-1},0)}}\\
 &= \ket{0}\bra{0} \otimes (\sum_{r_n}|\gamma^{(\R{n-1},r_n)}|^2) \ket{b_{00}^{(\R{n-1},0)}}\bra{b_{00}^{(\R{n-1},0)}} \\
 & \quad + \ket{1}\bra{1} \otimes (\sum_{r_n}|\gamma^{(\R{n-1},r_n)}|^2) \ket{b_{11}^{(\R{n-1},0)}}\bra{b_{11}^{(\R{n-1},0)}}.
\end{align*}

Here we prove a claim on the relationship between the square moduli of two $n \times m$ complex matrices.

\begin{claim} \label{lemma:polar_decomp}
 Let $M_{m,n}$ denote the set of $m \times n$ complex matrices.  If two linear operators $T \in M_{m,n}$ and $T' \in M_{m',n}$ satisfy
\begin{equation} \label{eq:lem_tdagt}
 T^\dag T = {T'}^\dag T',
\end{equation}
then there exists an isometry $U$ such that
\[
 T' = UT.
\]
\end{claim}

\begin{proof}
We define
\begin{equation} \label{eq:lem_H}
 H = T^\dag T.
\end{equation}
$T$ has a polar decomposition
\begin{equation} \label{eq:lem_PDT}
 T = U \sqrt{H},
\end{equation}
where $U$ is an isometry satisfying
\[
 U^\dag U = \I_{n \times n}.
\]
Here, $\I_{n \times n}$ denotes the $n \times n$ identity operator.  Eqs.\ (\ref{eq:lem_tdagt}) and (\ref{eq:lem_H}) imply that there exists an isometry $U'$ such that
\begin{equation*}
  T' = U' \sqrt{H},
\end{equation*}
where $U'$ is an isometry satisfying
\[
 {U'}^\dag U' = \I_{n \times n}.
\]
We define a new isometry
\[
 V = U'U^\dag.
\]
We see that
\[
 T' = VT
\]
because
\begin{eqnarray*}
 V T &=&U' U^\dag T = U' U^\dag U \sqrt{H} \\
 &=& U' \cdot \I_{n \times n} \cdot \sqrt{H} = U' \sqrt{H} = T'.
\end{eqnarray*}
\end{proof}

Using Claim \ref{lemma:polar_decomp} and setting $g^{\R{n-1}}$ to be
\begin{equation} \label{gRn-1}
 g^{\R{n-1}} = \sqrt{\sum_{r_n}|\gamma^{(\R{n-1},r_n)}|^2},
\end{equation}
we see that there exists an isometry $U_B^{\R{n-1}}$ such that
\begin{eqnarray}
 B^{\R{n-1}} &= &U_B^{\R{n-1}} \cdot (\ket{0}\bra{0} \otimes g^{\R{n-1}} \bra{b_{00}^{(\R{n-1},0)}} \nonumber \\
 &+& \ket{1}\bra{1} \otimes g^{\R{n-1}} \bra{b_{11}^{(\R{n-1},0)}}).
\end{eqnarray}
By introducing the following notation
\begin{eqnarray} \label{eq:LOCCimp_bllRn-1}
\bra{b_{00}^{\R{n-1}}} &=&g^{\R{n-1}} \bra{b_{00}^{(\R{n-1},0)}}, \\
\bra{b_{11}^{\R{n-1}}} &=& g^{\R{n-1}} \bra{b_{11}^{(\R{n-1},0)}}, \label{eq:LOCCimp_bllRn-1b}
\end{eqnarray}
$B^{\R{n-1}}$ is given by
\begin{equation} \label{eq:LOCCimp_BRn-1}
 B^{\R{n-1}} = U_B^{\R{n-1}} \cdot (\ket{0}\bra{0} \otimes \bra{b_{00}^{\R{n-1}}} + \ket{1}\bra{1} \otimes \bra{b_{11}^{\R{n-1}}}).
\end{equation}
We also define
\begin{eqnarray} \label{eq:LOCCimp_BRn-2.2}
 \bra{b_{00}^{\R{n-2}}} = \bra{b_{00}^{(\R{n-2},0)}} ,\label{eq:LOCCimp_BRn-2.2a}\\
 \bra{b_{11}^{\R{n-2}}} = \bra{b_{11}^{(\R{n-2},0)}}. \label{eq:LOCCimp_BRn-2.2b}
\end{eqnarray}
Eq.\ (\ref{eq:LOCCimp_onB2}) implies that
\begin{eqnarray}
 B^{(\R{n-2},0)\dag} B^{(\R{n-2},0)}  = B^{(\R{n-2},r_{n-1})\dag} B^{(\R{n-2},r_{n-1})}.  \label{eq:LOCCimp_BB}
\end{eqnarray}
Using Eq. (\ref{eq:LOCCimp_BRn-1}), Eq. (\ref{eq:LOCCimp_BB}) is equivalent to
\begin{align}
  \ket{b_{00}^{\R{n-2}}}\bra{b_{00}^{\R{n-2}}} &= \ket{b_{00}^{(\R{n-2},r_{n-1})}}\bra{b_{00}^{(\R{n-2},r_{n-1})}},  \label{eq:LOCCimp_bb00}\\
  \ket{b_{11}^{\R{n-2}}}\bra{b_{11}^{\R{n-2}}} &= \ket{b_{11}^{(\R{n-2},r_{n-1})}}\bra{b_{11}^{(\R{n-2},r_{n-1})}}.  \label{eq:LOCCimp_bb11}
\end{align}
For Eqs.(\ref{eq:LOCCimp_bb00}) and (\ref{eq:LOCCimp_bb11}) to hold, there must be phase factors
$\exp[i \varphi^{(\R{n-2},r_{n-1})}_0]$ and $\exp[i \varphi^{(\R{n-2},r_{n-1})}_1]$
such that
\begin{align}
 \label{eq:LOCCimp_bRn-2.1} \bra{b_{00}^{\R{n-2}}} &= \exp[i \varphi^{(\R{n-2},r_{n-1})}_0] \bra{b_{00}^{(\R{n-2},r_{n-1})}},\\
  \label{eq:LOCCimp_bRn-2.2} \bra{b_{11}^{\R{n-2}}} &= \exp[i \varphi^{(\R{n-2},r_{n-1})}_1]\bra{b_{11}^{(\R{n-2},r_{n-1})}}.
\end{align}

On the other hand, the left hand side of Eq.\ (\ref{eq:Blocks}) can be re-expressed as follows
\begin{align*}
 \bra{a_{kk}^{\R{n}}} X \ket{b_{ll}^{\R{n}*}} &= \gamma^{\R{n}} \bra{a_{kk}^{\R{n-1}}} X \ket{b_{ll}^{(\R{n-1},0)*}} \\
&= \frac{\gamma^{\R{n}}}{g^{\R{n-1}}} \bra{a_{kk}^{\R{n-1}}} X \ket{b_{ll}^{\R{n-1}*}} \\
&= \frac{\gamma^{\R{n}}}{g^{\R{n-1}}} \cdot e^{i \varphi^{(\R{n-2},r_{n-1})}_l} \bra{a_{kk}^{\R{n-1}}} X \ket{b_{ll}^{\R{n-2}*}}.
\end{align*}
By defining
\[
 \kappa_l^{\R{n}} = \frac{g^{\R{n-1}}e^{-i \varphi_{l}^{\R{n-1}}}c^{\R{n}}}{\gamma^{\R{n}}}
\]
and
\[
 \Eps{kl} = \bra{kl} U_{\theta} \ket{kl},
\]
Eq.\ (\ref{eq:Blocks}) is equivalent to
\begin{equation} \label{eq:LOCCimp_*}
 \bra{a_{kk}^{\R{n-1}}} X \ket{b_{ll}^{\R{n-2}*}} = \kappa_l^{\R{n}} \Eps{kl}.
\end{equation}
Eq.\ (\ref{eq:LOCCimp_*}) indicates that
\[
 \kappa_l^{(\R{n-1},r_{n})} = \kappa_l^{(\R{n-1},r'_{n})}
\]
for any $r_n$ and $r'_n$.  We define
\[
 \kappa_l^{\R{n-1}} = \kappa_l^{(\R{n-1},0)}.
\]
By Eq.\ (\ref{eq:LOCCimp_onA2}), we see that
\begin{multline} \label{eq:LOCCimp_doublecircle}
 A^{\R{n-2}\dag}A^{\R{n-2}} = \sum_{r_{n-1}} \ket{0}\bra{0} \otimes \ket{a_{00}^{(\R{n-2},r_{n-1})}}\bra{a_{00}^{(\R{n-2},r_{n-1})}} \\
 + \ket{1}\bra{1} \otimes \ket{a_{11}^{(\R{n-2},r_{n-1})}}\bra{a_{11}^{(\R{n-2},r_{n-1})}}.
\end{multline}
Notice that
\begin{equation} \label{eq:nit}
 \sum_{r_{n-1}} \ket{a_{00}^{(\R{n-2},r_{n-1})}}\bra{a_{00}^{(\R{n-2},r_{n-1})}}
\end{equation}
and 
\[ 
 \sum_{r_{n-1}} \ket{a_{11}^{(\R{n-2},r_{n-1})}}\bra{a_{11}^{(\R{n-2},r_{n-1})}}
\]
are positive semidefinite operators.  We may therefore define $2 \times 2$ matrices $A^{\R{n-2}}_{00}$ and $A^{\R{n-2}}_{11}$ by
\begin{align}
 &A^{\R{n-2}}_{00} = \sqrt{\sum_{r_{n-1}} \ket{a_{00}^{(\R{n-2},r_{n-1})}}\bra{a_{00}^{(\R{n-2},r_{n-1})}}}, \\
& A^{\R{n-2}}_{11} = \sqrt{\sum_{r_{n-1}} \ket{a_{11}^{(\R{n-2},r_{n-1})}}\bra{a_{11}^{(\R{n-2},r_{n-1})}}},
\end{align}
analogous to Eq.(\ref{fullA}).
Using Eq.\ (\ref{eq:LOCCimp_*}) and from the definition of $A_{kk}^{\R{n-1}}$, we see that
\begin{multline} \label{eq:LOCCimp_tau1}
\bra{b_{00}^{\R{n-2}*}}X (A_{00}^{\R{n-2}})^2 X\ket{b_{00}^{\R{n-2}*}}\\
 =\sum_{r_{n-{1}}} |\kappa_0^{(\R{n-2},r_{n-1})}|^2 \\
 = \bra{b_{00}^{\R{n-2}*}}X (A_{11}^{\R{n-2}})^2 X\ket{b_{00}^{\R{n-2}*}}
\end{multline}
\begin{multline} \label{eq:LOCCimp_tau2} 
\bra{b_{11}^{\R{n-2}*}}X (A_{00}^{\R{n-2}})^2 X\ket{b_{11}^{\R{n-2}*}} \\
 =\sum_{r_{n-{1}}} |\kappa_1^{(\R{n-2},r_{n-1})}|^2 \\
 = \bra{b_{11}^{\R{n-2}*}}X (A_{11}^{\R{n-2}})^2 X\ket{b_{11}^{\R{n-2}*}}  
\end{multline}
and
\begin{multline} \label{eq:LOCCimp_tau3}
\bra{b_{00}^{\R{n-2}*}}X (A_{00}^{\R{n-2}})^2 X\ket{b_{11}^{\R{n-2}*}} \\
 = \sum_{r_{n-{1}}} (\kappa_0^{(\R{n-2},r_{n-1})})^* \kappa_1^{(\R{n-2},r_{n-1})} \\
 = e^{-i \theta} \bra{b_{00}^{\R{n-2}*}}X (A_{11}^{\R{n-2}})^2 X\ket{b_{11}^{\R{n-2}*}}.
\end{multline}
In the definition of $\kappa$ we see that the $l$ dependence occurs only in a phase, and in particular
\[
 |\kappa_0^{\R{n}}| = \left| \left(\frac{g^{\R{n-1}}}{\gamma^{\R{n}}}\right)^2 c^{\R{n}} \right| = |\kappa_1^{\R{n}}|,
\]
therefore Eqs.\ (\ref{eq:LOCCimp_tau1}) and (\ref{eq:LOCCimp_tau2}) imply
\begin{multline} \label{eq:LOCCimp_A0b0=A0b1}
 \bra{b_{00}^{\R{n-2}*}} X (A_{00}^{\R{n-2}})^2 X \ket{b_{00}^{\R{n-2}*}}, \\
 = \bra{b_{11}^{\R{n-2}*}} X (A_{00}^{\R{n-2}})^2 X \ket{b_{11}^{\R{n-2}*}}.
\end{multline}

For two vectors
\[
A_{00}^{\R{n-2}} X\ket{b_{00}^{\R{n-2}*}}
\]
 and
\[
A_{00}^{\R{n-2}} X\ket{b_{11}^{\R{n-2}*}},
\]
the Cauchy-Schwarz inequality requires that
\begin{multline} \label{eq:LOCCimp_sigma}
 \bra{b_{00}^{\R{n-2}*}}X (A_{00}^{\R{n-2}})^2 X\ket{b_{00}^{\R{n-2}*}} {\times} \\
 \bra{b_{11}^{\R{n-2}*}}X (A_{00}^{\R{n-2}})^2 X\ket{b_{11}^{\R{n-2}*}} \\
 - |\bra{b_{00}^{\R{n-2}*}}X (A_{00}^{\R{n-2}})^2 X\ket{b_{11}^{\R{n-2}*}}|^2 \geq 0.
\end{multline}
By Eqs.\ (\ref{eq:LOCCimp_tau1}), (\ref{eq:LOCCimp_tau2}), and (\ref{eq:LOCCimp_tau3}), we also have that
\begin{multline} \label{eq:LOCCimp_sigma'}
 \bra{b_{00}^{\R{n-2}*}}X (A_{11}^{\R{n-2}})^2 X\ket{b_{00}^{\R{n-2}*}} {\times}\\
 \bra{b_{11}^{\R{n-2}*}}X (A_{11}^{\R{n-2}})^2 X\ket{b_{11}^{\R{n-2}*}} \\
 - |\bra{b_{00}^{\R{n-2}*}}X (A_{11}^{\R{n-2}})^2 X\ket{b_{11}^{\R{n-2}*}}|^2 \geq 0.
\end{multline}
Notice that the equality condition of Eq.\ (\ref{eq:LOCCimp_sigma}) holds if and only if the equality condition of Eq.\ (\ref{eq:LOCCimp_sigma'}) holds.  From this, we conclude that $A_{00}^{\R{n-2}} X\ket{b_{00}^{\R{n-2}*}}$ and $A_{00}^{\R{n-2}} X\ket{b_{11}^{\R{n-2}*}}$ are linearly independent if and only if $A_{11}^{\R{n-2}} X\ket{b_{00}^{\R{n-2}*}}$ and $A_{11}^{\R{n-2}} X\ket{b_{11}^{\R{n-2}*}}$ are linearly independent.

Recall that $X$ is full rank and observe that $\ket{b_{00}^{\R{n-2}*}}$ and $\ket{b_{11}^{\R{n-2}*}}$ are linearly independent, which can be seen by invoking Eqs.\ (\ref{eq:LOCCimp_BRn-2.2a}) and (\ref{eq:LOCCimp_BRn-2.2b}) to derive
\begin{align*}
 &\ket{b_{00}^{\R{n-2}*}} \propto \ket{b_{00}^{(\R{n-2},0)*}} \propto \ket{b_{00}^{(\R{n-2},0,0)*}}\\
 &\ket{b_{11}^{\R{n-2}*}} \propto \ket{b_{11}^{(\R{n-2},0)*}} \propto \ket{b_{11}^{(\R{n-2},0,0)*}}.
\end{align*}
Now, if $A_{00}^{\R{n-2}}$ is rank 1, then $A_{00}^{\R{n-2}} X\ket{b_{00}^{\R{n-2}*}}$ and $A_{00}^{\R{n-2}} X\ket{b_{11}^{\R{n-2}*}}$ are collinear, which means that so are $A_{11}^{\R{n-2}} X\ket{b_{00}^{\R{n-2}*}}$ and $A_{11}^{\R{n-2}} X\ket{b_{11}^{\R{n-2}*}}$.  Therefore, $A_{11}^{\R{n-2}}$ must also be rank 1.  Switching the roles of $A_{00}^{\R{n-2}}$ and $A_{11}^{\R{n-2}}$, we see that $A_{11}^{\R{n-2}}$ is rank 1 if and only if $A_{00}^{\R{n-2}}$ is rank 1.  Thus we arrive at the following relation
\begin{equation}
 \rank{A_{00}^{\R{n-2}}} = \rank{A_{11}^{\R{n-2}}} ,
\end{equation}
and thus by Eq.~(\ref{eq:LOCCimp_onA3}) we have 
\begin{equation} \label{equalrankA}
 \rank{A_{00}^{\R{n-3}}} = \rank{A_{11}^{\R{n-3}}} .
\end{equation}


\subsection{The proof for  $\rank{B_{00}^{\R{n-3}}}= \rank{B_{11}^{\R{n-3}}}$ in case $\rank{A_{kk}^{\R{n-3}}}=2$} \label{subsec:rankBequal}

In this section, we derive conditions on the rank of the block elements of the accumulated operators ${B_{ll}^{\R{n-3}}}$ in case the rank of $A_{00}^{\R{n-3}}$ is 2. 

First, using Eqs.\ (\ref{b00})-(\ref{eq:LOCCimp_BRn-2.2b}), the operators $B_{00}^{\R{n-3}}$ and $B_{11}^{\R{n-3}}$ defined in Eq.\ (\ref{fullB}) can be rewritten as
\begin{align*}
 B_{00}^{\R{n-3}} &= \sqrt{\sum_{r_{n-2}} \ket{b_{00}^{(\R{n-3},{r_{n-2}})}}\bra{b_{00}^{(\R{n-3},{r_{n-2}})}}},\\
 B_{11}^{\R{n-3}} &= \sqrt{\sum_{r_{n-2}} \ket{b_{11}^{(\R{n-3},{r_{n-2}})}}\bra{b_{11}^{(\R{n-3},{r_{n-2}})}}},
\end{align*}
and
\begin{align}
 \label{Herm1} & \tr{B}_{00}^{\R{n-3}} = B_{00}^{\R{n-3}*} = \sqrt{\sum_{r_{n-2}} \ket{b_{00}^{(\R{n-3},r_{n-2})*}}\bra{b_{00}^{(\R{n-3},r_{n-2})*}}} \\
 \label{Herm2} & \tr{B}_{11}^{\R{n-3}} = B_{11}^{\R{n-3}*} = \sqrt{\sum_{r_{n-2}} \ket{b_{11}^{(\R{n-3},r_{n-2})*}}\bra{b_{11}^{(\R{n-3},r_{n-2})*}}}.
\end{align}

To proceed further, let us prove in detail the following claim which will be used several times in what follows.
\begin{claim} \label{lem:eeff}
 Given two linearly inpdependent vectors $\ket{e_1}$ and $\ket{e_2}$, any two vectors $\ket{f_1}$ and $\ket{f_2}$ such that
\begin{align}
 \label{eq:eeff1} \norm{\ket{e_1}} &= \norm{\ket{f_1}} \\
 \label{eq:eeff2} \norm{\ket{e_2}} &= \norm{\ket{f_2}} \\
 \label{eq:eeff3} \braket{e_1}{e_2} &= e^{-i \theta} \braket{f_1}{f_2}
\end{align}
must be in the form
\begin{align}
 \ket{f_1} &= u \ket{e_1}\\
 \ket{f_2} &= e^{i \theta} u \ket{e_2}.
\end{align}
\end{claim}

\begin{proof}
 Eq.\ (\ref{eq:eeff1}) implies that there exists a unitary operator $u_1$ such that
\[
 \ket{f_1} = u_1 \ket{e_1},
\]
while Eq.\ (\ref{eq:eeff2}) implies that there exists a unitary operator $u_2$ such that
\[
 \ket{f_2} = u_2 \ket{e_2},
\]
Defining
\[
 u'_2 = e^{-i \theta} u_2,
\]
we have
\[
 \ket{f_2} = e^{i\theta} u'_2 \ket{e_2}
\]
and
\[
 e^{-i \theta} \braket{f_1}{f_2} = \bra{e_1}u^\dag_1 u'_2 \ket{e_2}.
\]

Let us define $\ket{f'_2}$ by
\[
 \ket{f'_2} = e^{-i \theta} \ket{f_2}.
\]
Because $\ket{e_1}$ and $\ket{e_2}$ are linearly independent, there exists a linear operator $T$ such that
\begin{align}
 \ket{f_1} &= T \ket{e_1} \\
 \ket{f'_2} &= T \ket{e_2}.
\end{align}

Consider any $\ket{\varphi}$ given as a linear combination of $\ket{e_1}$ and $\ket{e_2}$, \ie
\begin{equation} \label{eq:varphi_e1e2}
 \ket{\varphi} = \alpha \ket{e_1} + \beta \ket{e_2}.
\end{equation}
Notice that Eqs.\ (\ref{eq:eeff1}), (\ref{eq:eeff2}), and (\ref{eq:eeff3}) imply
\begin{align}
 \label{eq:eeff'1} \braket{e_1}{e_1} &= \braket{f_1}{f_1} \\
 \label{eq:eeff'2} \braket{e_2}{e_2} &= \braket{f'_2}{f'_2} \\
 \label{eq:eeff'3} \braket{e_1}{e_2} &= \braket{f_1}{f'_2}.
\end{align}
We see that $T$ must preserve the inner product of any two vectors of the form Eq.\ (\ref{eq:varphi_e1e2}), because
\begin{eqnarray}
 \bra{\varphi'} T^\dag T \ket{\varphi} &= (\alpha'^{*} \bra{e_1} + \beta'^{*} \bra{e_2}) T^\dag \cdot T (\alpha \ket{e_1} + \beta \ket{e_2}) \nonumber 
 \\
 &= \alpha'^{*} \alpha \braket{f_1}{f_1} + \alpha'^{*} \beta \braket{f_1}{f'_2} \nonumber \\
 & + {\beta'}^* \alpha \braket{f'_2}{f_1} + {\beta'}^* \beta \braket{f'_2}{f'_2}\\
 &= \alpha'^{*} \alpha \braket{e_1}{e_1} + \alpha'^{*} \beta \braket{e_1}{e_2} \nonumber \\
 &+ {\beta'}^* \alpha \braket{e_2}{e_1} + {\beta'}^* \beta \braket{e_2}{e_2}\\
 &= \braket{\varphi'}{\varphi}.
\end{eqnarray}
Therefore, $T$ is a unitary operator, which we can write as $u$, proving
\begin{align*}
 \ket{f_1} &= u \ket{e_1}\\
 \ket{f_2} &= e^{i \theta} u \ket{e_2}.
\end{align*}
\end{proof}

By applying Claim \ref{lem:eeff} to Eqs.\ (\ref{eq:LOCCimp_tau1}), (\ref{eq:LOCCimp_tau2}), and (\ref{eq:LOCCimp_tau3}), we see that there exists a $2 \times 2$ unitary operator $u^{\R{n-2}}$ such that
\begin{align}
 \label{eq:LOCCimp_omega1} A_{11}^{\R{n-2}} X \ket{b_{00}^{\R{n-2}*}} &= u^{\R{n-2}} A_{00}^{\R{n-2}} X \ket{b_{00}^{\R{n-2}*}},\\
 \label{eq:LOCCimp_omega2} A_{11}^{\R{n-2}} X \ket{b_{11}^{\R{n-2}*}} &= e^{i\theta} u^{\R{n-2}} A_{00}^{\R{n-2}} X \ket{b_{11}^{\R{n-2}*}}.
\end{align}

Since the $A_{kk}^{\R{n-2}}=A_{kk}^{\R{n-3}}$ are full rank by hypothesis, we can define
\[
 T^{\R{n-2}} = u^{\R{n-2}\dag} A_{11}^{\R{n-2}} (A_{00}^{\R{n-2}})^{-1},
\]
Eqs.\ (\ref{eq:LOCCimp_omega1}) and (\ref{eq:LOCCimp_omega2}) can be transformed to
\begin{align}
 T^{\R{n-2}} \cdot A_{00}^{\R{n-2}} X \ket{b_{00}^{\R{n-2}*}} &= A_{00}^{\R{n-2}} X \ket{b_{00}^{\R{n-2}*}},\\
 T^{\R{n-2}} \cdot A_{00}^{\R{n-2}} X \ket{b_{11}^{\R{n-2}*}} &= e^{i\theta} A_{00}^{\R{n-2}} X \ket{b_{11}^{\R{n-2}*}}.
\end{align}
This implies that $T^{\R{n-2}}$ has eigenvalues of $1$ and $e^{i\theta}$, and corresponding eigenvectors are $A_{00}^{\R{n-2}} X \ket{b_{00}^{\R{n-2}*}}$ and $A_{00}^{\R{n-2}} X \ket{b_{11}^{\R{n-2}*}}$, respectively.

Then $T^{\R{n-2}}$ must have the following eigen-decomposition,
\begin{equation} \label{eq:LOCCimp_xi1}
 T^{\R{n-2}} = S^{\R{n-2}} \begin{pmatrix} 1 & 0\\ 0 & e^{i\theta} \end{pmatrix} (S^{\R{n-2}})^{-1},
\end{equation}
where
\[
 S^{\R{n-2}} = A_{00}^{\R{n-2}} X \ket{b_{00}^{\R{n-2}*}}\bra{0} + A_{00}^{\R{n-2}} X \ket{b_{11}^{\R{n-2}*}}\bra{1}.
\]

Let us consider the singular value decomposition of $T^{\R{n-2}}$,
\begin{equation} \label{eq:LOCCimp_xi2}
 T^{\R{n-2}} = Q^{\R{n-2}} \begin{pmatrix} \sqrt{\lambda_1^{\R{n-2}}} & 0\\ 0 & \sqrt{\lambda_2^{\R{n-2}}} \end{pmatrix} R^{\R{n-2}},
\end{equation}
where $Q^{\R{n-2}}$ and $R^{\R{n-2}}$ are unitary operators and $\sqrt{\lambda_1^{\R{n-2}}}$ and $\sqrt{\lambda_2^{\R{n-2}}}$ are positive real numbers.
Equating Eqs.\ (\ref{eq:LOCCimp_xi1}) and (\ref{eq:LOCCimp_xi2}) and taking the determinant of both sides provide us a relationship
\begin{eqnarray*}
& \det S^{\R{n-2}} \cdot \det \begin{pmatrix} 1 & 0 \\  0& e^{i\theta} \end{pmatrix} \cdot \det (S^{\R{n-2}})^{-1}  \nonumber \\
& = \det Q^{\R{n-2}} \cdot \det \begin{pmatrix} \sqrt{\lambda_1^{\R{n-2}}} & 0 \\ 0 & \sqrt{\lambda_2^{\R{n-2}}} \end{pmatrix} \cdot \det R^{\R{n-2}}.
\end{eqnarray*}
Taking the absolute value of the left hand side, we have
\begin{eqnarray}
 |\det S^{\R{n-2}}| \cdot |\det \begin{pmatrix} 1 & 0\\ 0 & e^{i\theta} \end{pmatrix}| \cdot |\det (S^{\R{n-2}})^{-1}|  \nonumber \\
 = |\det S^{\R{n-2}}| \cdot | e^{i\theta} | \cdot \frac{1}{|\det (S^{\R{n-2}})|} = 1. \label{eq:det1}
\end{eqnarray}
Taking the absolute value of the right hand side, recalling that the determinant of a unitary is a phase, we have
\begin{eqnarray}
& |\det Q^{\R{n-2}}| \cdot |\det \begin{pmatrix} \sqrt{\lambda_1^{\R{n-2}}} & 0 \\ 0& \sqrt{\lambda_2^{\R{n-2}}} \end{pmatrix}| \cdot |\det R^{\R{n-2}}| \nonumber \\
 & = 1 \cdot \sqrt{\lambda_1^{\R{n-2}}}\sqrt{\lambda_2^{\R{n-2}}} \cdot 1. \label{eq:det2}
\end{eqnarray}
Equating Eqs. (\ref{eq:det1}) and (\ref{eq:det2}), we obtain
\[
 \sqrt{\lambda_1^{\R{n-2}}}\sqrt{\lambda_2^{\R{n-2}}} = 1.
\]
We see that the singular value decomposition of $T^{\R{n-2}}$ is restricted to
\begin{equation*}
 T^{\R{n-2}} = Q^{\R{n-2}} \begin{pmatrix} \sqrt{\lambda^{\R{n-2}}} & 0 \\ 0 & \frac{1}{\sqrt{\lambda^{\R{n-2}}}} \end{pmatrix} R^{\R{n-2}}.
\end{equation*}
We define
\begin{equation} \label{tky}
 {A'}_{00}^{\R{n-2}} = R^{\R{n-2}} A_{00}^{\R{n-2}}
\end{equation}
and
\begin{equation} 
{A'}_{11}^{\R{n-2}} = Q^{\R{n-2}\dag} u^{\R{n-2}} A_{11}^{\R{n-2}}.
\end{equation}

Using Eq.(\ref{eq:LOCCimp_onA3}), let us set 
\begin{equation} \label{eq:LOCCimp_AkkRn-3}
 {A'}_{00}^{(\R{n-3},0)} = \tilde{A}_{00}^{\R{n-3}} \mathrm{~and~}  {A'}_{11}^{(\R{n-3},0)} = \tilde{A}_{11}^{\R{n-3}}.
\end{equation}
Denoting $\lambda^{(\R{n-3},0)}$ by $\lambda^{\R{n-3}}$, we have
\begin{equation} \label{eq:LOCCimp_ARn-3}
 \tilde{A}_{11}^{\R{n-3}} = \begin{pmatrix} \sqrt{\lambda^{\R{n-3}}} & 0\\ 0 & \frac{1}{\sqrt{\lambda^{\R{n-3}}}} \end{pmatrix} \cdot \tilde{A}_{00}^{\R{n-3}}.
\end{equation}
Now the following relations hold
\begin{align*}
 A_{00}^{(\R{n-3},r_{n-2})\dag} A_{00}^{(\R{n-3},r_{n-2})} &= \tilde{A}_{00}^{\R{n-3}\dag} \tilde{A}_{00}^{\R{n-3}},\\
 A_{11}^{(\R{n-3},r_{n-2})\dag} A_{11}^{(\R{n-3},r_{n-2})} &= \tilde{A}_{11}^{\R{n-3}\dag} \tilde{A}_{11}^{\R{n-3}}.
\end{align*}

We can rewrite Eqs.\ (\ref{eq:LOCCimp_tau1}), (\ref{eq:LOCCimp_tau2}), and (\ref{eq:LOCCimp_tau3}) as
\begin{multline*}
\bra{b_{00}^{\R{n-2}*}}X \tilde{A}_{00}^{\R{n-3}\dag} \cdot \tilde{A}_{00}^{\R{n-3}} X\ket{b_{00}^{\R{n-2}*}} \\
 = \bra{b_{00}^{\R{n-2}*}}X \tilde{A}_{11}^{\R{n-3}\dag} \cdot \tilde{A}_{11}^{\R{n-3}} X\ket{b_{00}^{\R{n-2}*}},
\end{multline*}
\begin{multline*}
\bra{b_{11}^{\R{n-2}*}}X \tilde{A}_{00}^{\R{n-3}\dag} \cdot \tilde{A}_{00}^{\R{n-3}} X\ket{b_{11}^{\R{n-2}*}} \\
 = \bra{b_{11}^{\R{n-2}*}}X \tilde{A}_{11}^{\R{n-3}\dag} \cdot \tilde{A}_{11}^{\R{n-3}} X\ket{b_{11}^{\R{n-2}*}}
\end{multline*}
and
\begin{multline*}
\bra{b_{00}^{\R{n-2}*}}X \tilde{A}_{00}^{\R{n-3}\dag} \cdot \tilde{A}_{00}^{\R{n-3}} X\ket{b_{11}^{\R{n-2}*}}\\
= e^{-i \theta} \bra{b_{00}^{\R{n-2}*}}X \tilde{A}_{11}^{\R{n-3}\dag} \cdot \tilde{A}_{11}^{\R{n-3}} X\ket{b_{11}^{\R{n-2}*}},
\end{multline*}
which are equivalent to
\begin{multline} \label{eq:LOCCimp_theta1}
 \bra{b_{00}^{\R{n-2}*}}X \tilde{A}_{00}^{\R{n-3}\dag} \cdot \tilde{A}_{00}^{\R{n-3}} X\ket{b_{00}^{\R{n-2}*}}  \\
  = \bra{b_{00}^{\R{n-2}*}}X \tilde{A}_{00}^{\R{n-3}\dag} \cdot \begin{pmatrix} \lambda^{\R{n-3}} & 0\\ 0 & \frac{1}{\lambda^{\R{n-3}}} \end{pmatrix}
  \\ \times \tilde{A}_{00}^{\R{n-3}} X\ket{b_{00}^{\R{n-2}*}},
\end{multline}  
\begin{multline} \label{eq:LOCCimp_theta2}
 \bra{b_{11}^{\R{n-2}*}}X \tilde{A}_{00}^{\R{n-3}\dag} \cdot \tilde{A}_{00}^{\R{n-3}} X\ket{b_{11}^{\R{n-2}*}}  \\
  = \bra{b_{11}^{\R{n-2}*}}X \tilde{A}_{00}^{\R{n-3}\dag} \cdot \begin{pmatrix} \lambda^{\R{n-3}} & 0 \\ 0 & \frac{1}{\lambda^{\R{n-3}}} \end{pmatrix}\\
  \times \tilde{A}_{00}^{\R{n-3}} X\ket{b_{11}^{\R{n-2}*}},
\end{multline}
and
\begin{multline} \label{eq:LOCCimp_theta3}
\bra{b_{00}^{\R{n-2}*}}X \tilde{A}_{00}^{\R{n-3}\dag} \cdot \tilde{A}_{00}^{\R{n-3}} X\ket{b_{11}^{\R{n-2}*}}\\
   = e^{-i \theta} \bra{b_{00}^{\R{n-2}*}}X \tilde{A}_{00}^{\R{n-3}\dag} \cdot \begin{pmatrix} \lambda^{\R{n-3}} & 0\\ 0 & \frac{1}{\lambda^{\R{n-3}}} \end{pmatrix} \\
   \times \tilde{A}_{00}^{\R{n-3}} X\ket{b_{11}^{\R{n-2}*}}.
\end{multline}

We now consider two cases, namely, when $\lambda^{\R{n-3}} \neq 1$ and $\lambda^{\R{n-3}} = 1$, and prove that $\rank{B_{00}^{\R{n-3}}}= \rank{B_{11}^{\R{n-3}}}$ holds for each case.

\subsubsection{The $\lambda^{\R{n-3}} \neq 1$ case} \label{lneq1}

We consider the case of $\lambda^{\R{n-3}} \neq 1$.
Let us define the coefficients of $\tilde{A}_{00}^{\R{n-3}} X\ket{b_{00}^{\R{n-2}*}}$ and $\tilde{A}_{00}^{\R{n-3}} X\ket{b_{11}^{\R{n-2}*}}$ by 
\begin{align*}
 \tilde{A}_{00}^{\R{n-3}} X\ket{b_{00}^{\R{n-2}*}} &= a_0^{\R{n-2}} e^{i \omega_{00}^{\R{n-2}}} \ket{0} + b_0^{\R{n-2}} e^{i \omega_{01}^{\R{n-2}}} \ket{1} \\
 \tilde{A}_{00}^{\R{n-3}} X\ket{b_{11}^{\R{n-2}*}} &= a_1^{\R{n-2}} e^{i \omega_{10}^{\R{n-2}}} \ket{0} + b_1^{\R{n-2}} e^{i \omega_{11}^{\R{n-2}}} \ket{1},
\end{align*}
where $a_0^{\R{n-2}}$, $a_1^{\R{n-2}}$, $b_0^{\R{n-2}}$, $b_1^{\R{n-2}}$, $\omega_{00}^{\R{n-2}}$, $\omega_{01}^{\R{n-2}}$, $\omega_{10}^{\R{n-2}}$, and $\omega_{11}^{\R{n-2}}$ are all real numbers.  With this notation, Eqs.\ (\ref{eq:LOCCimp_theta1}), (\ref{eq:LOCCimp_theta2}), and (\ref{eq:LOCCimp_theta3}) read
\begin{multline} \label{eq:LOCCimp_theta'1}
 (a_0^{\R{n-2}})^2 + (b_0^{\R{n-2}})^2 \\
 = \lambda^{\R{n-3}}(a_0^{\R{n-2}})^2 + \frac{1}{\lambda^{\R{n-3}}}(b_0^{\R{n-2}})^2,
\end{multline} 
\begin{multline} \label{eq:LOCCimp_theta'2}
 (a_1^{\R{n-2}})^2 + (b_1^{\R{n-2}})^2\\
 = \lambda^{\R{n-3}}(a_1^{\R{n-2}})^2 + \frac{1}{\lambda^{\R{n-3}}}(b_1^{\R{n-2}})^2,  
\end{multline} 
 and
\begin{multline} \label{eq:LOCCimp_theta'3}
 a_0^{\R{n-2}} a_1^{\R{n-2}} e^{i (\omega_{00}^{\R{n-2}}-\omega_{10}^{\R{n-2}})} \\
 + b_0^{\R{n-2}} b_1^{\R{n-2}} e^{i (\omega_{01}^{\R{n-2}}-\omega_{11}^{\R{n-2}})} \\
  = e^{-i \theta}(\lambda^{\R{n-3}}a_0^{\R{n-2}} a_1^{\R{n-2}} e^{i (\omega_{00}^{\R{n-2}}-\omega_{10}^{\R{n-2}})} \\ 
+ \frac{1}{\lambda^{\R{n-3}}}b_0^{\R{n-2}} b_1^{\R{n-2}} e^{i (\omega_{10}^{\R{n-2}}-\omega_{11}^{\R{n-2}})}).  
\end{multline}

Eq.\ (\ref{eq:LOCCimp_A0b0=A0b1}) becomes
\begin{equation} \label{eq:LOCCimp_varphi1}
 (a_0^{\R{n-2}})^2 + (b_0^{\R{n-2}})^2 = (a_1^{\R{n-2}})^2 + (b_1^{\R{n-2}})^2.
\end{equation}
Eqs.\ (\ref{eq:LOCCimp_theta'1}) and (\ref{eq:LOCCimp_theta'2}) imply that
\begin{equation} \label{eq:LOCCimp_varphi2}
 \frac{(a_0^{\R{n-2}})^2}{(b_0^{\R{n-2}})^2} = \frac{\frac{1}{\lambda^{\R{n-3}}}-1}{1-\lambda^{\R{n-3}}} = \frac{1}{\lambda^{\R{n-3}}} = \frac{(a_1^{\R{n-2}})^2}{(b_1^{\R{n-2}})^2},
\end{equation}
where the assumption $\lambda^{\R{n-3}} \neq 1$ guarantees that the quotient is well defined.  The last two equations imply that
\begin{equation} \label{eq:LOCCimp_varphi3}
 a_0^{\R{n-2}} = a_1^{\R{n-2}} \mathrm{~and~}  b_0^{\R{n-2}} = b_1^{\R{n-2}}.
\end{equation}
We divide both sides of Eq.\ (\ref{eq:LOCCimp_theta'3}) by $b_0^{\R{n-2}} b_1^{\R{n-2}} e^{i (\omega_{01}^{\R{n-2}}-\omega_{11}^{\R{n-2}})}$ and use Eqs.\ (\ref{eq:LOCCimp_varphi2}) and (\ref{eq:LOCCimp_varphi3}) to derive
\begin{align}
 &\frac{1}{\lambda^{\R{n-3}}} \cdot e^{-i \delta^{\R{n-2}}} + 1 = e^{-i \theta}(e^{-i \delta^{\R{n-2}}} + \frac{1}{\lambda^{\R{n-3}}}) \nonumber \\
 &\Longleftrightarrow \lambda^{\R{n-3}} + e^{-i \delta^{\R{n-2}}}= e^{-i \theta}(1 + \lambda^{\R{n-3}} \cdot e^{-i \delta^{\R{n-2}}}) \label{eq:LOCCimp_psi},
\end{align}
where
\[
 \delta^{\R{n-2}}= \omega_{01}^{\R{n-2}}-\omega_{11}^{\R{n-2}}-\omega_{00}^{\R{n-2}}+\omega_{10}^{\R{n-2}}.
\]
Eq.\ (\ref{eq:LOCCimp_psi}) is equivalent to
\begin{equation} \label{eq:LOCCimp_lambdadelta}
 e^{-i \delta^{\R{n-2}}} = \frac{\lambda^{\R{n-3}} - e^{-i \theta}}{\lambda^{\R{n-3}} e^{-i \theta}-1},
\end{equation}
which implies independence of $r_{n-2}$:
\[
 e^{-i \delta^{(\R{n-3},0)}} = e^{-i \delta^{(\R{n-3},r_{n-2})}}. 
\]
By defining
\[
 \delta^{\R{n-3}} = \delta^{(\R{n-3},0)} \mathrm{~and~} \Delta^{\R{n-2}}= \omega_{01}^{\R{n-2}}-\omega_{00}^{\R{n-2}}
\]
and using Eqs.\ (\ref{eq:LOCCimp_varphi2}) and (\ref{eq:LOCCimp_varphi3}), we obtain
\begin{multline*}
 \tilde{A}_{00}^{\R{n-3}} X \ket{b_{00}^{\R{n-2}*}} \\
 = a_0^{\R{n-2}}e^{i \omega_{00}^{\R{n-2}}}(\ket{0} + \sqrt{\lambda^{\R{n-3}}}e^{i \Delta^{\R{n-2}}}\ket{1})
\end{multline*} 
and
\begin{multline*}
 \tilde{A}_{00}^{\R{n-3}} X \ket{b_{11}^{\R{n-2}*}} \\
 = a_0^{\R{n-2}}e^{i \omega_{10}^{\R{n-2}}}(\ket{0} + \sqrt{\lambda^{\R{n-3}}}e^{i (\Delta^{\R{n-2}}-\delta^{\R{n-2}})}\ket{1}).
\end{multline*}
Introducing
\[
 c_0^{\R{n-3}} = a_0^{\R{n-2}}e^{i \omega_{00}^{\R{n-2}}} \mathrm{~and~} c_1^{\R{n-3}} = a_0^{\R{n-2}}e^{i \omega_{10}^{\R{n-2}}},
\]
we can write these relations as
\begin{multline*}
 \tilde{A}_{kk}^{\R{n-3}} X \ket{b_{ll}^{\R{n-2}*}}
 = c_l^{\R{n-3}} \cdot \begin{pmatrix} \sqrt{\lambda^{\R{n-3}}} & 0\\ 0 & \frac{1}{\sqrt{\lambda^{\R{n-3}}}} \end{pmatrix}^k \\
  \times \begin{pmatrix} 1 &0\\0  & e^{-i\delta^{\R{n-3}}} \end{pmatrix}^l \cdot (\ket{0} + \sqrt{\lambda^{\R{n-3}}}e^{i \Delta^{\R{n-2}}}\ket{1}).
\end{multline*}

For any $k,l,k',l' \in \{0,1\}$ the following holds
\begin{eqnarray} \label{eq:LOCCimp_individualtrace}
& \Tr \tilde{A}_{kk}^{\R{n-3}} X \ket{b_{ll}^{\R{n-2}*}}\bra{b_{ll}^{\R{n-2}*}} X \tilde{A}_{kk}^{\R{n-3}\dag} \nonumber \\
& = \Tr A_{k'k'}^{\R{n-3}} X \ket{b_{l'l'}^{\R{n-2}*}}\bra{b_{l'l'}^{\R{n-2}*}} X A_{k'k'}^{\R{n-3}\dag}.
\end{eqnarray}
This gives
\begin{multline}\label{eq:LOCCimp_DAXB}
 \tilde{A}_{00}^{\R{n-3}} X \tr{B}_{11}^{\R{n-3}}B_{11}^{\R{n-3}*} X \tilde{A}_{00}^{\R{n-3}\dag}\\
 = \sum_{r_{n-2}} \tilde{A}_{00}^{\R{n-3}} X \ket{b_{11}^{(\R{n-3},{r_{n-2}})*}}\bra{b_{11}^{(\R{n-3},{r_{n-2}})*}} X \tilde{A}_{00}^{\R{n-3}\dag}\\
 = \sum_{r_{n-2}} \begin{pmatrix} 1 & 0\\ 0& e^{-i\delta^{\R{n-3}}} \end{pmatrix}
 \cdot \tilde{A}_{00}^{\R{n-3}} X \ket{b_{11}^{(\R{n-3},{r_{n-2}})*}} {\times(c.c.)} \\
 =  \begin{pmatrix} 1 & 0 \\ 0 & e^{-i\delta^{\R{n-3}}} \end{pmatrix} \cdot \tilde{A}_{00}^{\R{n-3}} X \tr{B}_{00}^{\R{n-3}} \times (c.c.).
\end{multline}
Note that $\tilde{A}_{00}^{\R{n-3}}$, $\begin{pmatrix} 1 & 0\\0 & e^{-i\delta^{\R{n-3}}} \end{pmatrix}$, and $X$ are all full rank, implying that their determinants are nonzero.  Taking the determinant of both sides, we obtain
\begin{equation} \label{eq:LOCCimp_sigma2}
 \det \tr{B}_{00}^{\R{n-3}}B_{00}^{\R{n-3}*} = \det \tr{B_{11}^{\R{n-3}}}B_{11}^{\R{n-3}*}
\end{equation}
Because $B_{00}^{\R{n-3}}$ and $B_{11}^{\R{n-3}}$ are $2 \times 2$ nonzero matrices, if the rank of $B_{00}^{\R{n-3}}$ is one, then the left hand side of Eq.\ (\ref{eq:LOCCimp_sigma2}) is zero and thus the rank of $B_{00}^{\R{n-3}}$ is also one.  On the other hand, if $B_{00}^{\R{n-3}}$ is full rank, then so is $B_{00}^{\R{n-3}}$.  Therefore, we have that for any successful LOCC implementaion protocol,
\begin{equation} \label{equalrankBneq1}
 \rank{B_{00}^{\R{n-3}}} = \rank{B_{11}^{\R{n-3}}}
\end{equation}
for $\lambda^{\R{n-3}} \neq 1$.

\subsubsection{The $\lambda^{\R{n-3}} = 1$ case} \label{l=1}

Next, we consider the case when $\lambda^{\R{n-3}} = 1$. Notice that this means that $A_{11}^{\R{n-3}}(A_{00}^{\R{n-3}})^{-1}$ is a unitary operation, which is exactly the assumption of Lemma \ref{lem:dke}.  By Eq.(\ref{eq:LOCCimp_ARn-3}), we have
\begin{equation} \label{tildeAequal}
 \tilde{A}^{\R{n-3}}_{11} = \tilde{A}^{\R{n-3}}_{00}.
\end{equation}  
Eq.\ (\ref{eq:LOCCimp_theta3}) implies
\begin{multline*}
\bra{b_{00}^{\R{n-2}*}}X \tilde{A}_{00}^{\R{n-3}\dag} \cdot \tilde{A}_{00}^{\R{n-3}} X\ket{b_{11}^{\R{n-2}*}} \\
= e^{-i \theta} \bra{b_{00}^{\R{n-2}*}}X \tilde{A}_{11}^{\R{n-3}\dag} \cdot \tilde{A}_{11}^{\R{n-3}} X\ket{b_{11}^{\R{n-2}*}} \\
 = e^{-i \theta} \bra{b_{00}^{\R{n-2}*}}X \tilde{A}_{00}^{\R{n-3}\dag} \cdot \tilde{A}_{00}^{\R{n-3}} X\ket{b_{11}^{\R{n-2}*}}.
\end{multline*}
Because $e^{-i \theta} \neq 1$, the last equation forces
\begin{equation} \label{eq:LOCCimp_A0b0perpA0b1}
 \bra{b_{00}^{\R{n-2}*}}X \tilde{A}_{00}^{\R{n-3}\dag} \cdot \tilde{A}_{00}^{\R{n-3}} X\ket{b_{11}^{\R{n-2}*}} = 0.
\end{equation}
We define $\ket{g^{\R{n-2}}}$ and $\ket{g^{\R{n-2}\perp}}$ by
\begin{align*}
 \ket{g^{\R{n-2}}} &= \tilde{A}_{00}^{\R{n-3}} X\ket{b_{00}^{\R{n-2}*}},\\
 \ket{g^{\R{n-2}\perp}} &= \tilde{A}_{00}^{\R{n-3}} X\ket{b_{11}^{\R{n-2}*}}.
\end{align*}
By Eqs.\ (\ref{eq:LOCCimp_A0b0=A0b1}) and (\ref{eq:LOCCimp_A0b0perpA0b1}), we have
\[
 \norm{\ket{g^{\R{n-2}}}} = \norm{\ket{g^{\R{n-2}}}} \equiv C^{\R{n-2}}
\] 
and
\[ 
\braket{g^{\R{n-2}}}{g^{\R{n-2}\perp}} = 0.
\]
Because $\ket{g^{\R{n-2}}}$ and $\ket{g^{\R{n-2}\perp}}$ are two-dimensional vectors,
\[
 \ket{g^{\R{n-2}}}\bra{g^{\R{n-2}}} + \ket{g^{\R{n-2}\perp}}\bra{g^{\R{n-2}\perp}} = \I_2,
\]
where $\I_2$ is the $2 \times 2$ identity matrix.

Next, we obtain the following relation
\begin{align}
& \tilde{A}_{00}^{\R{n-3}} X \tr{B}_{00}^{\R{n-3}} B_{00}^{\R{n-3}*} X \tilde{A}_{00}^{\R{n-3}\dag} \nonumber \\
& \qquad + \tilde{A}_{00}^{\R{n-3}} X \tr{B}_{11}^{\R{n-3}} B_{11}^{\R{n-3}*} X \tilde{A}_{00}^{\R{n-3}\dag} \nonumber \\
&= \sum_{r_{n-2}} \tilde{A}_{00}^{\R{n-3}} X \ket{b_{00}^{\R{n-2}*}}\bra{b_{00}^{\R{n-2}*}} X \tilde{A}_{00}^{\R{n-3}\dag} \nonumber \\ 
& \qquad + \tilde{A}_{00}^{\R{n-3}} X \ket{b_{11}^{\R{n-2}*}}\bra{b_{11}^{\R{n-2}*}} X \tilde{A}_{00}^{\R{n-3}\dag} \nonumber \\
&= \sum_{r_{n-2}} \ket{g^{\R{n-2}}}\bra{g^{\R{n-2}}} + \ket{g^{\R{n-2}\perp}}\bra{g^{\R{n-2}\perp}} \nonumber \\
&= [\sum_{r_{n-2}} (C^{(\R{n-3},r_{n-2})})^2] \cdot \I_2, \label{eq:LOCCimp_CI}
\end{align}
where we have used Eqs.\ (\ref{Herm1}) and (\ref{Herm2}).
Thus one sees that
\[
\tilde{A}_{00}^{\R{n-3}} X \tr{B}_{00}^{\R{n-3}} B_{00}^{\R{n-3}*} X \tilde{A}_{00}^{\R{n-3}\dag}
\]
and
\[
\tilde{A}_{00}^{\R{n-3}} X \tr{B}_{11}^{\R{n-3}} B_{11}^{\R{n-3}*} X \tilde{A}_{00}^{\R{n-3}\dag}
\]
commute, and are therefore simultaneously diagonalizable:
\begin{multline*}
 \tilde{A}_{00}^{\R{n-3}} X \tr{B}_{00}^{\R{n-3}} B_{00}^{\R{n-3}*} X \tilde{A}_{00}^{\R{n-3}\dag} \\
 = V^{\R{n-3}} \begin{pmatrix} e_0^{\R{n-3}} & \\ & f_0^{\R{n-3}} \end{pmatrix} V^{\R{n-3}\dag},
\end{multline*}
and
\begin{multline*}
\tilde{A}_{00}^{\R{n-3}} X \tr{B}_{11}^{\R{n-3}} B_{11}^{\R{n-3}*} X \tilde{A}_{00}^{\R{n-3}\dag} \\
 = V^{\R{n-3}} \begin{pmatrix} e_1^{\R{n-3}} & \\ & f_1^{\R{n-3}} \end{pmatrix} V^{\R{n-3}\dag}.
\end{multline*}
Also we note that
\begin{multline} \label{eq:LOCCimp_equaltr}  
 \Tr \tilde{A}_{00}^{\R{n-3}} X \tr{B}_{00}^{\R{n-3}} B_{00}^{\R{n-3}*} X \tilde{A}_{00}^{\R{n-3}\dag} \\
 = \sum_{r_{n-2}} \Tr \ket{g^{(\R{n-3},r_{n-2})}}\bra{g^{(\R{n-3},r_{n-2})}} \\
 = \sum_{r_{n-2}} \norm{\ket{g^{(\R{n-3},r_{n-2})}}}^2 \\
 = \sum_{r_{n-2}} \norm{\ket{g^{(\R{n-3},r_{n-2})\perp}}}^2 \\
 = \sum_{r_{n-2}} \Tr \ket{g^{(\R{n-3},r_{n-2})\perp}}\bra{g^{(\R{n-3},r_{n-2})\perp}} \\
 = \Tr \tilde{A}_{00}^{\R{n-3}} X \tr{B}_{11}^{\R{n-3}} B_{11}^{\R{n-3}*} X \tilde{A}_{00}^{\R{n-3}\dag}. 
\end{multline}

Eq.\ (\ref{eq:LOCCimp_CI}) implies that
\[
 e_0^{\R{n-3}} + e_1^{\R{n-3}} = f_0^{\R{n-3}} + f_1^{\R{n-3}},
\]
while Eq.\ (\ref{eq:LOCCimp_equaltr}) implies
\[
 e_0^{\R{n-3}} + f_0^{\R{n-3}} = e_1^{\R{n-3}} + f_1^{\R{n-3}}.
\]
Comparing the two equations, we obtain
\begin{align*}
 e_0^{\R{n-3}} &= f_1^{\R{n-3}}\\
 e_1^{\R{n-3}} &= f_0^{\R{n-3}}.
\end{align*}
The last two equations imply that
\begin{equation} \label{equalrankBeq1}
 \rank{B_{00}^{\R{n-3}}} = \rank{B_{11}^{\R{n-3}}},
\end{equation}
completing the proof.

\section{Proof of Lemma \ref{lem:rank1}}

In this section we provide the proof for Lemma \ref{lem:rank1}.  We begin with case (a).

As discussed in the proof of Lemma \ref{lem:equalranks}, $\bra{a_{kk}^{(\R{n-2},r_{n-1},r_n)}}$ does not have any $r_n$-dependence, which implies that 
\[
 A_{kk}^{\R{n-3}} = \sqrt{\sum_{r_{n-1}} \ket{a^{(\R{n-3},r_{n-2},r_{n-1},0)}_{kk}}\bra{a^{(\R{n-3},r_{n-2},r_{n-1},0)}_{kk}}}.
\]
For this to be rank 1, there must exist complex numbers ${\gamma'}_0^{(\R{n-2},r_{n-1})}$ and ${\gamma'}_0^{(\R{n-2},r_{n-1})}$ such that
\begin{align*}
 \bra{a_{00}^{(\R{n-2},r_{n-1},0)}} &= {\gamma'}_0^{(\R{n-2},r_{n-1})}\bra{a_{00}^{(\R{n-2},0,0)}}\\
 \bra{a_{11}^{(\R{n-2},r_{n-1},0)}} &= {\gamma'}_1^{(\R{n-2},r_{n-1})}\bra{a_{11}^{(\R{n-2},0,0)}}.
\end{align*}

By Eq.\ (\ref{eq:LOCCimp_tau1}), we see that
\[
 |{\gamma'}_0^{(\R{n-2},r_{n-1})}| = |{\gamma'}_1^{(\R{n-2},r_{n-1})}|,
\]
hence ${\gamma'}_0^{(\R{n-2},r_{n-1})}$ and ${\gamma'}_1^{(\R{n-2},r_{n-1})}$ can be rewritten as
\begin{align}
 {\gamma'}_0^{(\R{n-2},r_{n-1})} &= |{\gamma'}_0^{(\R{n-2},r_{n-1})}| \exp(i \arg {\gamma'}_0^{(\R{n-2},r_{n-1})}) \\
 {\gamma'}_1^{(\R{n-2},r_{n-1})} &= |{\gamma'}_0^{(\R{n-2},r_{n-1})}| \exp(i \arg {\gamma'}_1^{(\R{n-2},r_{n-1})}).
\end{align}

Defining a positive number
\[
 {g'}^{\R{n-2}} = \sqrt{\sum_{r_{n-1}} |{\gamma'}_0^{(\R{n-2},r_{n-1})}|^2},
\]
it is easy to see by Eqs.\ (\ref{eq:LOCCimp_onA3}) and (\ref{eq:ARn-3}) that
\begin{multline*}
 A^{\R{n-2}\dag} A^{\R{n-2}}\\
 = ({g'}^{\R{n-2}})^2 \cdot  (\ket{0}\bra{0} \otimes \ket{a_{00}^{(\R{n-2},0,0)}}\bra{a_{00}^{(\R{n-2},0,0)}}\\
 + \ket{1}\bra{1} \otimes \ket{a_{11}^{(\R{n-2},0,0)}}\bra{a_{11}^{(\R{n-2},0,0)}}).
\end{multline*}
Therefore, by Claim \ref{lemma:polar_decomp}, there exists an isometry $U_A^{\R{n-2}}$  such that
\begin{multline}
 A^{\R{n-2}} = U_A^{\R{n-2}} \cdot   (\ket{0}\bra{0} \otimes {g'}^{\R{n-2}} \bra{a_{00}^{(\R{n-2},0)}}\\
 + \ket{1}\bra{1} \otimes {g'}^{\R{n-2}} \bra{a_{11}^{(\R{n-2},0)}}).
\end{multline}

Using Eqs.\ (\ref{a00Rn-1}), (\ref{a11Rn-1}), and (\ref{ARn-1}), it can be checked by direct calculation that the random unitary operation $\{ p^{\R{n-1}},U^{\R{n-1}} \}$ in question is given by
\[
 p^{\R{n-1}} = \frac{|{\gamma'}_0^{(\R{n-2},r_{n-1})}|^2}{({g'}^{\R{n-2}})^2}
\]
and 
\begin{align*}
 U^{\R{n-1}} & =  \begin{pmatrix} \exp(i \arg {\gamma'}_0^{\R{n-1}}) &0\\0  & \exp(i \arg {\gamma'}_1^{\R{n-1}}) \end{pmatrix} \cdot U_A^{\R{n-2}\dag}.
\end{align*} 

Next, we proceed to case (b).  In Eq.\ (\ref{eq:LOCCimp_onB2}), the left hand side is independent of $r_{n-1}$ therefore the right hand should also be independent of $r_{n-1}$ and we set
\[
 B^{\R{n-2}} = B^{(\R{n-2},0)}.
\]
By Eq.\ (\ref{eq:LOCCimp_onB3}) we have
\begin{equation}
 B^{\R{n-3}\dag} B^{\R{n-3}} = \sum_{r_{n-2}} B^{\R{n-1}\dag} B^{\R{n-1}}.
\end{equation}
Combined with Eqs.\ (\ref{eq:LOCCimp_bllRn-1}) and (\ref{eq:LOCCimp_bllRn-1b}), we obtain 
\[
 B_{ll}^{\R{n-3}} = \sqrt{\sum_{r_{n-2}} \ket{b^{(\R{n-2},0)}_{ll}}\bra{b^{(\R{n-2},0)}_{ll}}}.
\]
Because $B_{00}^{\R{n-3}}$ and $B_{11}^{\R{n-3}}$ are both rank 1 by assumption, there exist two complex numbers ${\gamma''}_0^{(\R{n-3},r_{n-2})}$ and ${\gamma''}_1^{(\R{n-3},r_{n-2})}$ such that
\begin{align*}
 \bra{b_{00}^{(\R{n-2},0)}} &= {\gamma''}_0^{(\R{n-3},r_{n-2})}\bra{b_{00}^{(\R{n-3},0,0)}}\\
 \bra{b_{11}^{(\R{n-2},0)}} &= {\gamma''}_1^{(\R{n-3},r_{n-2})}\bra{b_{11}^{(\R{n-3},0,0)}}.
\end{align*}

By Eq.\ (\ref{eq:LOCCimp_A0b0=A0b1}), we see that
\[
 |{\gamma''}_0^{(\R{n-3},r_{n-2})}| = |{\gamma''}_1^{(\R{n-3},r_{n-2})}|.
\]

Defining a positive number
\[
 {g''}^{\R{n-3}} = \sqrt{\sum_{r_{n-2}} |{\gamma''}_0^{(\R{n-3},r_{n-2})}|^2},
\]
it is easy to see that
\begin{multline*}
 B^{\R{n-3}\dag} B^{\R{n-3}} = ({g''}^{\R{n-3}})^2\\
 \times (\ket{0}\bra{0} \otimes \ket{b_{00}^{(\R{n-3},0,0)}}\bra{b_{00}^{(\R{n-3},0,0)}}\\
 + \ket{1}\bra{1} \otimes \ket{b_{11}^{(\R{n-3},0,0)}}\bra{b_{11}^{(\R{n-3},0,0)}}).
\end{multline*}
Therefore, there exists an isometry $U_B^{\R{n-3}}$ by Claim \ref{lemma:polar_decomp} such that
\begin{multline*}
 B^{\R{n-3}} = U_B^{\R{n-3}} \cdot   (\ket{0}\bra{0} \otimes {g''}^{\R{n-3}} \bra{b_{00}^{(\R{n-3},0)}}\\
 + \ket{1}\bra{1} \otimes {g''}^{\R{n-3}} \bra{b_{11}^{(\R{n-3},0)}}).
\end{multline*}

It can be checked by direct calculation that the random unitary operation $\{ q^{\R{n-2}},V^{\R{n-2}} \}$ in question is given by
\[
 q^{\R{n-2}} = \frac{|{\gamma''}_0^{\R{n-2}}|^2}{({g''}^{\R{n-3}})^2}
\]
and 
\begin{align*}
 V^{\R{n-2}} =\begin{pmatrix} \exp(i \arg {\gamma''}_0^{\R{n-2}}) &0\\0  & \exp(i \arg {\gamma''}_1^{\R{n-2}}) \end{pmatrix} \cdot U_B^{\R{n-3}\dag}.
\end{align*}

\section{Proof of Lemma \ref{lem:whenrankA2}}

In this section we present the proof for Lemma \ref{lem:whenrankA2}.  The proof requires the following claim, which will also be proven in this section.

Let us define
\[
 E_{kl} = \tilde{A}_{kk}^{\R{n-3}} X \tr{{B'}}_{ll}^{\R{n-3}}.
\]
\begin{claim} \label{final_lemma}
When $A_{00}^{\R{n-3}}$ and $B_{00}^{\R{n-3}}$ are both rank 2 in a successful protocol, then the accumulated operators $A^{\R{n-3}}$ and $B^{\R{n-3}}$ can be expressed as follows
\begin{multline} \label{AinConv3}
 A^{\R{n-3}} = U_A^{\R{n-3}} \cdot (
 \ket{0}\bra{0} \otimes E_{00} (\tr{{B'}}_{00}^{\R{n-3}})^{-1} X^{-1}\\
 + \ket{1}\bra{1} \otimes E_{00} \begin{pmatrix} 1 & 0 \\ 0 & e^{-i {\delta'}^{\R{n-3}}} \end{pmatrix}  (\tr{{B'}}_{00}^{\R{n-3}})^{-1} X^{-1},
\end{multline}
and
\begin{multline} \label{BinConv3}
 B^{\R{n-3}} = U_B^{\R{n-3}} (\ket{0}\bra{0} \otimes {B'}_{00}^{\R{n-3}}\\
   + \ket{1}\bra{1} \otimes \begin{pmatrix} \sqrt{\mu^{\R{n-3}}} & 0 \\ 0 & \frac{1}{\sqrt{\mu^{\R{n-3}}}} \end{pmatrix}{B'}_{00}^{\R{n-3}})
\end{multline}
where $\mu^{\R{n-3}}$ and ${\delta'}^{\R{n-3}}$ satisfy
\begin{equation} \label{mud}
 \mu^{\R{n-3}} + e^{-i {\delta'}^{\R{n-3}}} = e^{i \theta} (1 + \mu^{\R{n-3}} e^{-i {\delta'}^{\R{n-3}}}).
\end{equation}
Moreoever, there exist phase factors $e^{i \varphi_0}$ and $e^{i \varphi_1}$ such that
\begin{multline} \label{phi0phi1}
 E_{00}^{\R{n-3}\dag} E_{00}^{\R{n-3}} = \frac{|e_{00}|^2 (1+\mu^{\R{n-3}})}{2} \times\\
   \left[\sum_{l=0,1}(\ket{0}+\sqrt{\mu^{\R{n-3}}}e^{-i \varphi_l}\ket{1})(\bra{0}+\sqrt{\mu^{\R{n-3}}}e^{i \varphi_l}\bra{1})\right].
\end{multline}
\end{claim}

\begin{proof}
As in the proof of lemma~\ref{lem:equalranks}, we divide the proof for the cases where $\lambda^{\R{n-3}} \neq 1$ and $\lambda^{\R{n-3}} = 1$, and start with the former.  Defining 
\[
 T_B^{\R{n-3}} = B_{11}^{\R{n-3}*} (B_{00}^{\R{n-3}*})^{-1},
\]
Eq.\ (\ref{eq:LOCCimp_sigma2}) is equivalent to
\begin{equation} \label{eq:LOCCimp_sigma2''}
 \det \tr{B}_{00}^{\R{n-3}}B_{00}^{\R{n-3}*} = \det \tr{B_{00}^{\R{n-3}}} \cdot T_B^{\R{n-3}\dag}T_B^{\R{n-3}} \cdot B_{00}^{\R{n-3}*} .
\end{equation}
By hypothesis, $\det B_{00}^{\R{n-3}}$ is nonzero.  Thus dividing Eq.\ (\ref{eq:LOCCimp_sigma2''}) by $\det B_{00}^{\R{n-3}*}$ and $\det \tr{B}_{00}^{\R{n-3}}$, we have
\begin{equation} \label{eq:LOCCimp_sigma2'}
 1 = \det T_B^{\R{n-3}\dag}T_B^{\R{n-3}}.
\end{equation}

Denoting the singular value decomposition of $T_B^{\R{n-3}}$ by
\[
 T_B^{\R{n-3}} = Q_B^{\R{n-3}} \begin{pmatrix} \sqrt{\mu_1^{\R{n-3}}} & 0\\ 0 & \sqrt{\mu_2^{\R{n-3}}} \end{pmatrix} R_B^{\R{n-3}},
\]
Eq.\ (\ref{eq:LOCCimp_sigma2'}) implies that
\[
 1 = \mu_1^{\R{n-3}} \mu_2^{\R{n-3}},
\]
which restricts the singular value decomposition of $T_B^{\R{n-3}}$ to be
\[
 T_B^{\R{n-3}} = Q_B^{\R{n-3}} \begin{pmatrix} \sqrt{\mu^{\R{n-3}}} &0 \\ 0 & \frac{1}{\sqrt{\mu^{\R{n-3}}}} \end{pmatrix} R_B^{\R{n-3}}.
\]

Let us define
\begin{equation} \label{tanomu}
 {B'}_{00}^{\R{n-3}} = R_B^{\R{n-3}*} B_{00}^{\R{n-3}} \mathrm{~and~}  {B'}_{11}^{\R{n-3}} = \tr{Q}^{\R{n-3}}  B_{11}^{\R{n-3}}.
\end{equation}
Recalling the definition of $T_B^{\R{n-3}}$, ${B'}_{00}^{\R{n-3}}$ and ${B'}_{11}^{\R{n-3}}$ have to satisfy
\begin{equation} \label{eq:LOCCimp_B'Rn-3}
 {B'}_{11}^{\R{n-3}} = \begin{pmatrix} \sqrt{\mu^{\R{n-3}}} & 0\\ 0 & \frac{1}{\sqrt{\mu^{\R{n-3}}}} \end{pmatrix} {B'}_{00}^{\R{n-3}}.
\end{equation} 

For any $k,l,k',l' \in \{0,1\}$, Eq.\ (\ref{eq:LOCCimp_individualtrace}) implies that
\begin{eqnarray*}
 &\Tr \tilde{A}_{kk}^{\R{n-3}} X \tr{B}_{ll}^{\R{n-3}}B_{ll}^{\R{n-3}*} X \tilde{A}_{kk}^{\R{n-3}\dag} \nonumber \\
 &= \Tr A_{k'k'}^{\R{n-3}} X \tr{B}_{l'l'}^{\R{n-3}}B_{l'l'}^{\R{n-3}*} X A_{k'k'}^{\R{n-3}\dag}.
\end{eqnarray*}
Using this relation, we have that
\begin{align} 
 & \Tr \tilde{A}_{kk}^{\R{n-3}} X \tr{{B'}}_{ll}^{\R{n-3}}{B'}_{ll}^{\R{n-3}*} X \tilde{A}_{kk}^{\R{n-3}\dag} \nonumber \\
 &= \Tr \tilde{A}_{kk}^{\R{n-3}} X \tr{B}_{ll}^{\R{n-3}}B_{ll}^{\R{n-3}*} X \tilde{A}_{kk}^{\R{n-3}\dag} \nonumber \\
 &= \Tr A_{k'k'}^{\R{n-3}} X \tr{B}_{l'l'}^{\R{n-3}}B_{l'l'}^{\R{n-3}*} X A_{k'k'}^{\R{n-3}\dag} \nonumber \\
 &= \Tr A_{k'k'}^{\R{n-3}} X \tr{{B'}}_{l'l'}^{\R{n-3}}{B'}_{l'l'}^{\R{n-3}*} X A_{k'k'}^{\R{n-3}\dag} \label{eq:LOCCimp_Tr}.
\end{align}
Eqs.\ (\ref{eq:LOCCimp_ARn-3}) and (\ref{eq:LOCCimp_B'Rn-3}) show that $\det \tilde{A}_{kk}^{\R{n-3}}$ and $\det {B'}_{kk}^{\R{n-3}}$ are independent of $k$, and so
\begin{align}
 & \det (\tilde{A}_{kk}^{\R{n-3}} X \tr{{B'}}_{ll}^{\R{n-3}}{B'}_{ll}^{\R{n-3}*} X \tilde{A}_{kk}^{\R{n-3}\dag}) \nonumber \\
 &= \det (\tilde{A}_{00}^{\R{n-3}} X \tr{{B'}}_{00}^{\R{n-3}}{B'}_{00}^{\R{n-3}*} X \tilde{A}_{00}^{\R{n-3}\dag}) \label{eq:LOCCimp_det}.
\end{align} 

Eqs.\ (\ref{eq:LOCCimp_ARn-3}) and (\ref{eq:LOCCimp_B'Rn-3}) imply
\begin{equation} \label{eq:LOCCimp_Emulambda}
 E_{kl} = \begin{pmatrix} \sqrt{\lambda^{\R{n-3}}} & 0\\0 & \frac{1}{\sqrt{\lambda^{\R{n-3}}}} \end{pmatrix}^k \cdot E_{00} \cdot \begin{pmatrix} \sqrt{\mu^{\R{n-3}}} & 0\\ 0 & \frac{1}{\sqrt{\mu^{\R{n-3}}}} \end{pmatrix}^l.
\end{equation}
By choosing $(k,l)=(0,0)$ and $(k',l')=(1,0)$ in Eq.\ (\ref{eq:LOCCimp_Tr}), we have
\begin{eqnarray*}
 &\Tr E_{00}E_{00}^\dag  = \Tr \begin{pmatrix} \sqrt{\lambda^{\R{n-3}}} & 0\\0 & \frac{1}{\sqrt{\lambda^{\R{n-3}}}} \end{pmatrix} \cdot E_{00}E_{00}^\dag \\
 &\cdot \begin{pmatrix} \sqrt{\lambda^{\R{n-3}}} & 0\\0 & \frac{1}{\sqrt{\lambda^{\R{n-3}}}} \end{pmatrix}.
\end{eqnarray*}
We define the elements of $E_{00}$ by
\[
 E_{00} = \begin{pmatrix} e_{00} & e_{01} \\ e_{10} & e_{11} \end{pmatrix},
\]
which gives
\begin{multline} \label{eq:LOCCimp_lambda}
 (|e_{00}|^2 + |e_{01}|^2) + (|e_{10}|^2 + |e_{11}|^2)  \\
 = \lambda^{\R{n-3}}(|e_{00}|^2 + |e_{01}|^2) + \frac{1}{\lambda^{\R{n-3}}} (|e_{10}|^2 + |e_{11}|^2),
\end{multline}
implying
\begin{equation}
 (1 - \lambda^{\R{n-3}}) (|e_{00}|^2 + |e_{01}|^2) = (\frac{1}{\lambda^{\R{n-3}}} - 1) (|e_{10}|^2 + |e_{11}|^2)
\end{equation}
which is equivalent to
\begin{equation}
 \lambda^{\R{n-3}} (|e_{00}|^2 + |e_{01}|^2) = (|e_{10}|^2 + |e_{11}|^2). \label{eq:LOCCimp_R1}
\end{equation}
For $(k,l)=(0,1)$ and $(k',l')=(1,1)$, we perform a similar calculation and obtain
\begin{multline} \label{eq:LOCCimp_R2}
\lambda^{\R{n-3}} (\mu^{\R{n-3}}|e_{00}|^2 + \frac{1}{\mu^{\R{n-3}}}|e_{01}|^2)\\
 = (\mu^{\R{n-3}}|e_{10}|^2 + \frac{1}{\mu^{\R{n-3}}}|e_{11}|^2).
\end{multline}

For $(k,l) = (0,0)$ and $(k',l') = (0,1)$, Eqs.\ (\ref{eq:LOCCimp_tau1}), (\ref{eq:LOCCimp_tau2}), (\ref{eq:LOCCimp_tau3}) and the cyclic property of trace imply
\begin{multline*}
 \Tr E_{00}^\dag E_{00} = \Tr \begin{pmatrix} \sqrt{\mu^{\R{n-3}}} & 0\\0 & \frac{1}{\sqrt{\mu^{\R{n-3}}}} \end{pmatrix} \\
 \cdot E_{00}^\dag E_{00} \cdot \begin{pmatrix} \sqrt{\mu^{\R{n-3}}} & 0\\0 & \frac{1}{\sqrt{\mu^{\R{n-3}}}} \end{pmatrix},
\end{multline*}
from which we obtain
\begin{multline*}
 (|e_{00}|^2 + |e_{10}|^2) + (|e_{01}|^2 + |e_{11}|^2)\\
 = \mu^{\R{n-3}} (|e_{00}|^2 + |e_{10}|^2) + \frac{1}{\mu^{\R{n-3}}}(|e_{01}|^2 + |e_{11}|^2) 
\end{multline*}
that is equivalent to
\begin{equation}
 \label{eq:LOCCimp_R3} \mu^{\R{n-3}} (|e_{00}|^2 + |e_{10}|^2) = (|e_{01}|^2 + |e_{11}|^2).
\end{equation}
For $(k,l) = (1,0)$ and $(k',l') = (1,1)$, we apply similar calculation, resulting in
\begin{multline} \label{eq:LOCCimp_R4}
 \mu^{\R{n-3}} (\lambda^{\R{n-3}} |e_{00}|^2 + \frac{1}{\lambda^{\R{n-3}}} |e_{10}|^2)\\
 = (\lambda^{\R{n-3}} |e_{01}|^2 + \frac{1}{\lambda^{\R{n-3}}} |e_{11}|^2).
\end{multline}

Multiplying Eq.\ (\ref{eq:LOCCimp_R4}) by $\lambda^{\R{n-3}}$ and subtracting it from Eq.\ (\ref{eq:LOCCimp_R3}), we obtain
\begin{equation} \label{eq:LOCCimp_alpha'}
  |e_{01}|^2 = \mu^{\R{n-3}} |e_{00}|^2.
\end{equation}
Multiplying Eq.\ (\ref{eq:LOCCimp_R4}) by $\lambda^{\R{n-3}}$ and Eq.\ (\ref{eq:LOCCimp_R3}) by $(\lambda^{\R{n-3}})^2$, and subtracting the former from the latter, we obtain
\begin{equation} \label{eq:LOCCimp_beta'}
  \mu^{\R{n-3}}|e_{10}|^2 = |e_{11}|^2.
\end{equation}
Using Eqs.\ (\ref{eq:LOCCimp_R1}), (\ref{eq:LOCCimp_alpha'}) , and (\ref{eq:LOCCimp_beta'}), we have
\begin{equation} \label{eq:LOCCimp_gamma'}
  \lambda^{\R{n-3}}|e_{00}|^2 = |e_{10}|^2.
\end{equation}
Therefore, we conclude from Eqs.\ (\ref{eq:LOCCimp_alpha'}), (\ref{eq:LOCCimp_beta'}), and (\ref{eq:LOCCimp_gamma'}) that
\[
 E_{00} = |e_{00}|^2 \cdot \begin{pmatrix} e^{i w} & \sqrt{\mu^{\R{n-3}}} e^{i x}\\ \sqrt{\lambda^{\R{n-3}}} e^{i y} & \sqrt{\lambda^{\R{n-3}} \mu^{\R{n-3}}} e^{i z} \end{pmatrix},
\]
for some real numbers $w$, $x$, $y$, and $z$.

Eq.\ (\ref{eq:LOCCimp_DAXB}) and the definition of ${B'}_{00}^{\R{n-3}}$ and ${B'}_{11}^{\R{n-3}}$ give us
\[
 E_{01}E_{01}^\dag = \begin{pmatrix} 1 & 0\\0 & e^{-i\delta^{\R{n-3}}} \end{pmatrix} E_{00}E_{00}^\dag \begin{pmatrix} 1 & 0\\0 & e^{i\delta^{\R{n-3}}} \end{pmatrix}.
\]
Comparing the $(0,1)$-element of the two sides, we have
\begin{multline*}
\mu^{\R{n-3}} \sqrt{\lambda^{\R{n-3}}} e^{i (w-y)} + \sqrt{\lambda^{\R{n-3}}} e^{i (x-z)}\\
  = e^{i \delta^{\R{n-3}}} (\sqrt{\lambda^{\R{n-3}}} e^{i (w-y)} + \mu^{\R{n-3}} \sqrt{\lambda^{\R{n-3}}} e^{i (x-z)}).
\end{multline*}
Defining $\Delta' = -w + x + y -z$, the last equation is equivalent to
\[
 e^{i \delta^{\R{n-3}}} = \frac{1+\mu^{\R{n-3}}e^{i \Delta'}}{\mu^{\R{n-3}}+e^{i \Delta'}}.
\]
Using Eq.\ (\ref{eq:LOCCimp_lambdadelta}), we have
\[
 \frac{\lambda^{\R{n-3}} - e^{i \theta}}{\lambda^{\R{n-3}} e^{i \theta} - 1} = \frac{1+\mu^{\R{n-3}}e^{i \Delta'}}{\mu^{\R{n-3}}+e^{i \Delta'}}.
\]
It can be easily checked that this equation is equivalent to
\[
 \frac{\mu^{\R{n-3}} - e^{i \theta}}{\mu^{\R{n-3}} e^{i \theta} - 1} = \frac{1+\lambda^{\R{n-3}}e^{i \Delta'}}{\lambda^{\R{n-3}}+e^{i \Delta'}}.
\]
The denominator and the numerator of the left hand side have the same magnitude, implying that there exists ${\delta'}^{\R{n-3}}$ such that
\begin{equation} \label{eq:LOCCimp_mudelta}
 \frac{\mu^{\R{n-3}} - e^{i \theta}}{\mu^{\R{n-3}} e^{i \theta} - 1} = e^{i {\delta'}^{\R{n-3}}}.
\end{equation}

It can be shown that $\tr{E}_{00} E_{00}^*$ and $\tr{E}_{10} E_{10}^*$ satisfy
\begin{equation} \label{eq:LOCCimp_Edelta'}
 \tr{E}_{10} E_{10}^* = \begin{pmatrix} 1 & 0\\ 0& e^{-i{\delta'}^{\R{n-3}}} \end{pmatrix} \tr{E}_{00} E_{00}^* \begin{pmatrix} 1 & 0\\0 & e^{i{\delta'}^{\R{n-3}}} \end{pmatrix},
\end{equation}
which can be seen as follows.  First, we observe that
\begin{equation}
 \tr{E}_{00} E_{00}^* = \left(\begin{array}{cc} W_{00} & W_{01} \\W_{10} & W_{11}\end{array}\right)
\end{equation}
and 
\begin{equation}
 \tr{E}_{10} E_{10}^*  = \left(\begin{array}{cc} W_{00} & W'_{01} \\W'_{10} & W_{11}\end{array}\right)
\end{equation}
where the elements $W_{kl}$ and $W'_{kl}$ are given by 
\begin{align*}
W_{00} &= 1+\lambda^{\R{n-3}}, \\
W_{01} &= \sqrt{\mu^{\R{n-3}}} \left( e^{i (w-x)} + \lambda^{\R{n-3}} e^{i (y-z)} \right ),  \\
W_{10} &= \sqrt{\mu^{\R{n-3}}} \left(e^{-i (w-x)} + \lambda^{\R{n-3}}  e^{-i (y-z)} \right ), \\
W_{11} &= \mu^{\R{n-3}}(1+\lambda^{\R{n-3}}),\\
W'_{01} &= \sqrt{\mu^{\R{n-3}}} \left( \lambda^{\R{n-3}} e^{i (w-x)} +  e^{i (y-z)} \right ),  \\
W'_{10} &= \sqrt{\mu^{\R{n-3}}} \left(\lambda^{\R{n-3}} e^{-i (w-x)} +   e^{-i (y-z)} \right ).
\end{align*}

Now,
\begin{multline*}
 \lambda^{\R{n-3}} \sqrt{\mu^{\R{n-3}}} e^{i (w-x)} + \sqrt{\mu^{\R{n-3}}} e^{i (y-z)} \\
 = (\sqrt{\mu^{\R{n-3}}} e^{i (w-x)} + \lambda^{\R{n-3}} \sqrt{\mu^{\R{n-3}}} e^{i (y-z)}) \\
 \times \frac{\lambda^{\R{n-3}} \sqrt{\mu^{\R{n-3}}} e^{i (w-x)} + \sqrt{\mu^{\R{n-3}}} e^{i (y-z)}}{\sqrt{\mu^{\R{n-3}}} e^{i (w-x)} + \lambda^{\R{n-3}} \sqrt{\mu^{\R{n-3}}} e^{i (y-z)}} \\
= (\sqrt{\mu^{\R{n-3}}} e^{i (w-x)} + \lambda^{\R{n-3}} \sqrt{\mu^{\R{n-3}}} e^{i (y-z)}) \\
 \times \frac{\mu^{\R{n-3}}-e^{i\theta}}{\mu^{\R{n-3}}e^{i\theta}-1} \\
 = (\sqrt{\mu^{\R{n-3}}} e^{i (w-x)} + \lambda^{\R{n-3}} \sqrt{\mu^{\R{n-3}}} e^{i (y-z)}) \times e^{i {\delta'}^{\R{n-3}}},
\end{multline*}
where the last equality follows from Eq.\ (\ref{eq:LOCCimp_mudelta}).  Therefore, Eq.\ (\ref{eq:LOCCimp_Edelta'}) holds.

Finally, using Eq.\ (\ref{eq:LOCCimp_Emulambda}) we have
\begin{multline*} 
 \tr{E}_{kl} E_{kl}^* = \begin{pmatrix} \sqrt{\mu^{\R{n-3}}} & 0\\0 & \frac{1}{\sqrt{\mu^{\R{n-3}}}} \end{pmatrix}^l \begin{pmatrix} 1 & 0\\0 & e^{-i{\delta'}^{\R{n-3}}} \end{pmatrix}^k \\
\times  \tr{E}_{00} E_{00}^* \begin{pmatrix} 1 & 0\\ 0 & e^{i{\delta'}^{\R{n-3}}} \end{pmatrix}^k \begin{pmatrix} \sqrt{\mu^{\R{n-3}}} & 0\\0 & \frac{1}{\sqrt{\mu^{\R{n-3}}}} \end{pmatrix}^l
\end{multline*}
and
\[
 \mu^{\R{n-3}} + e^{-i {\delta'}^{\R{n-3}}} = e^{i \theta} (1 + \mu^{\R{n-3}} e^{-i {\delta'}^{\R{n-3}}}),
\]
which follows directly from Eq.\ (\ref{eq:LOCCimp_mudelta}).

Using Claim \ref{lemma:polar_decomp} with Eqs.\ (\ref{tanomu}) and (\ref{eq:LOCCimp_B'Rn-3}), we prove that, for the case $\lambda^{\R{n-3}} \neq 1$, Eq.\ (\ref{BinConv3}) holds.  On the other hand, using Claim \ref{lemma:polar_decomp} with Eqs.\ (\ref{eq:LOCCimp_onA3}), (\ref{tky}), (\ref{eq:LOCCimp_AkkRn-3}), (\ref{eq:LOCCimp_Edelta'}), we see that Eq.\ (\ref{AinConv3}) holds.  Finally, because
\[
 \left|\frac{2 \sqrt{\mu^{\R{n-3}}} e^{i (x-w)}}{1+\lambda^{\R{n-3}}} + \frac{2 \lambda^{\R{n-3}}\sqrt{\mu^{\R{n-3}}} e^{i (z-y)}}{1+\lambda^{\R{n-3}}} \right| \leq 2 \sqrt{\mu^{\R{n-3}}},
\]
there exist phase factors $e^{i \varphi_0}$ and $e^{i \varphi_1}$ such that Eq.\ (\ref{phi0phi1}) holds.  This completes the proof of the lemma for the case $\lambda^{\R{n-3}} \neq 1$.

Next, let us assume $\lambda^{\R{n-3}} = 1$.  In this case, we have
\begin{equation}
 \label{klm1} \tilde{A}_{00}^{\R{n-3}} X \tr{B}_{00}^{\R{n-3}} = V^{\R{n-3}} \begin{pmatrix} \sqrt{e_0^{\R{n-3}}} & 0\\0 & \sqrt{f_0^{\R{n-3}}} \end{pmatrix} W_0^{\R{n-3}*}
\end{equation}
and
\begin{multline}
\label{klm2} \tilde{A}_{00}^{\R{n-3}} X \tr{B}_{11}^{\R{n-3}} \\
 = \tilde{A}_{00}^{\R{n-3}} X \tr{B}_{00}^{\R{n-3}} \cdot \tr{\ W}_0^{\R{n-3}} \begin{pmatrix} \sqrt{\mu^{\R{n-3}}} & 0\\ 0& \frac{1}{\sqrt{\mu^{\R{n-3}}}} \end{pmatrix} W_1^{\R{n-3}*},
\end{multline}
where $W_0^{\R{n-3}}$ and $W_1^{\R{n-3}}$ are some unitary operators and $\mu^{\R{n-3}} = f_0^{\R{n-3}} / e_0^{\R{n-3}}$, which is well defined because $B_{00}^{\R{n-3}}$ is assumed to be full rank.

We define
\begin{align}
 \label{B'toB:0} {B'}_{00}^{\R{n-3}} &= W_0^{\R{n-3}} B_{00}^{\R{n-3}} \\
 \label{B'toB:1} {B'}_{11}^{\R{n-3}} &= W_1^{\R{n-3}} B_{11}^{\R{n-3}},
\end{align}
which together with Eq.\ (\ref{klm2}) give us
\begin{equation} \label{unf1}
 {B'}_{11}^{\R{n-3}} = \begin{pmatrix} \sqrt{\mu^{\R{n-3}}} & 0\\ 0 & \frac{1}{\sqrt{\mu^{\R{n-3}}}} \end{pmatrix}{B'}_{00}^{\R{n-3}}
\end{equation}
and 
\begin{equation} \label{eq:unf2}
 {B'}_{ll}^{\R{n-3}\dag}{B'}_{ll}^{\R{n-3}} = B_{ll}^{\R{n-3}\dag}B_{ll}^{\R{n-3}}.
\end{equation}
Using Claim \ref{lemma:polar_decomp} with Eqs.\ (\ref{unf1}) and (\ref{eq:unf2}), we prove that, for the case $\lambda^{\R{n-3}}=1$, Eq.\ (\ref{BinConv3}) holds.

We also define
\[
 E_{kl}^{\R{n-3}} = V^{\R{n-3}\dag}\tilde{A}_{kk}^{\R{n-3}} X \tr{{B'}}_{ll}^{\R{n-3}}.
\]
It is easy to see that
\begin{equation} \label{kih}
 \tr{E_{10}}E_{10}^* = \begin{pmatrix} 1 & 0 \\ 0 & e^{-i{\delta'}^{\R{n-3}}} \end{pmatrix} \tr{E_{00}}E_{00}^* \begin{pmatrix} 1 & 0 \\ 0 & e^{i{\delta'}^{\R{n-3}}} \end{pmatrix}
\end{equation}
for any $e^{i{\delta'}^{\R{n-3}}}$, but we choose $e^{i{\delta'}^{\R{n-3}}}$ so that
\[
 \mu^{\R{n-3}} + e^{-i{\delta'}^{\R{n-3}}} = e^{-i \theta} (1 + \mu^{\R{n-3}}e^{-i{\delta'}^{\R{n-3}}}).
\]
Using Claim \ref{lemma:polar_decomp} with Eqs.\ (\ref{eq:LOCCimp_onA3}), (\ref{tky}), (\ref{eq:LOCCimp_AkkRn-3}), (\ref{eq:unf2}), and (\ref{kih}), we see that, for the case $\lambda^{\R{n-3}}=1$, Eq.\ (\ref{AinConv3}) holds.
Finally, it suffices to choose
\[
 e^{i \varphi_0} = 1 \mathrm{~and~} e^{i \varphi_1} = -1
\]
for Eq.\ (\ref{phi0phi1}).

\end{proof}

In Ref.~\cite{Andersson}, an algorithm is presented to construct a measurement operation such that an accumulated operator $T$ `splits' into $T'$ and $T''$, where
\begin{equation}
 T^\dag T = {T'}^\dag T' + {T''}^\dag T''.
\end{equation}
Therefore, we can find a two-outcome measurement $\{ {M'}^{(r_{n-2}|\R{n-3})} \}_{r_{n-2} = 0,1}$ such that
\begin{multline}
 {M'}^{(r_{n-2}|\R{n-3})} A^{\R{n-3}}= \sqrt{\frac{|e_{00}|^2 (1+\mu^{\R{n-3}})}{2}} \\
 \times \left\{\ket{0}\bra{0} \otimes (\bra{0}+\sqrt{\mu^{\R{n-3}}}e^{i \varphi_{r_{n-2}}}\bra{1})  \right.\\
 \left. + \ket{1}\bra{1} \otimes (\bra{0}+\sqrt{\mu^{\R{n-3}}}e^{i (\varphi_{r_{n-2}}-{\delta'}^{\R{n-3}})}\bra{1})\right\}\\
 \times \left[ I \otimes \tr{({B'}^{\R{n-3}*}X)}^{-1} \right].
\end{multline}

Bob's new measurement operation $\{{K'}^{(r_{n-1}|\R{n-2})}\}_{r_{n-1}=0,1}$ is given by
\begin{multline*}
 K^{(r_{n-1}|\R{n-2})} = (\ket{0}_{B,\IN}\bra{0} \otimes \bra{r_{n-1}}_{B,\res}\\
 + \ket{1}_{B,\IN}\bra{1} \otimes \bra{r_{n-1}}_{B,\res} e^{i \varphi_{r_{n-2}}} \cdot {u'}^{\R{n-3}\dag}) \cdot U_B^{\R{n-3}\dag},
\end{multline*}
where ${u'}^{\R{n-3}\dag}$ is a unitary operator that satisfies
\begin{multline}
 \begin{pmatrix} 1 & 0\\ 0 & e^{-i {\delta'}^{\R{n-3}}} \end{pmatrix}^k \times \begin{pmatrix} \sqrt{\mu^{\R{n-3}}} & 0 \\ 0 & \frac{1}{\sqrt{\mu^{\R{n-3}}}} \end{pmatrix}^l \\
 \times (\ket{0}_{B,\res}+\sqrt{\mu^{\R{n-3}}}e^{i \varphi_x}\ket{1}_{B,\res}) = \\
 \Eps{kl} ({u'}^{\R{n-3}})^l \times \begin{pmatrix} 1 & 0 \\ 0 & e^{-i {\delta'}^{\R{n-3}}} \end{pmatrix}^k (\ket{0}_{B,\res}+\sqrt{\mu^{\R{n-3}}}e^{i \varphi_x}\ket{1}_{B,\res}),
\end{multline}
whose existence is guaranteed by Eq.\ (\ref{mud}) and Claim~\ref{lem:eeff}.

Alice's final measurement operation $\{ {M'}^{(r_n|\R{n-1})} \}_{r_n} \}$ is now given by
\begin{align}
 \{ {M'}^{(r_n|\R{n-2},0)} \} &= \{ \I_{A,\IN} \}\\
 \{ {M'}^{(r_n|\R{n-2},1)} \} &= \begin{pmatrix} 1 & 0 \\ 0 & e^{i {\delta'}^{\R{n-3}}} \end{pmatrix}_{A,\IN}.
\end{align}
This completes the construction of the new measurements required in Lemma~\ref{lem:whenrankA2}.

\section{Proof of Lemma \ref{lem:dke}}

Finally, it can be seen that Eqs.~(\ref{fullA}), (\ref{eq:nit}), (\ref{tky}), and (\ref{eq:LOCCimp_CI}) imply Eq.~(\ref{eq:dnk}), proving Lemma~\ref{lem:dke}.


\end{document}